\documentclass[11pt,reqno]{article}

\usepackage{amssymb}
\usepackage{latexsym}
\usepackage{amsmath}
\usepackage{calrsfs}

\usepackage{hyperref}

\usepackage[noadjust]{cite}

\usepackage{amsthm}
\usepackage{empheq}
\usepackage{bm}
\usepackage[capitalise]{cleveref}

\usepackage{todonotes}

\usepackage[margin=2.8cm]{geometry}
\numberwithin{equation}{section}
\theoremstyle{definition}
\newtheorem{theorem}{Theorem}[section]
\newtheorem{corollary}[theorem]{Corollary}
\newtheorem{proposition}[theorem]{Proposition}
\newtheorem{definition}[theorem]{Definition}

\newtheorem{notation}[theorem]{Notation}
\newtheorem{remark}[theorem]{Remark}
\newtheorem{lemma}[theorem]{Lemma}

\newcounter{alp}
\newcounter{ara}
\newcounter{rom}

\newenvironment{alphalist}{\begin{list}{(\alph{alp})\hfill}{\usecounter{alp}
     \topsep0ex \labelwidth.6cm \leftmargin.6cm \labelsep0cm
     \rightmargin0cm \parsep0ex \itemsep0ex
     \partopsep0ex}}{\end{list}}

\newcommand{\numberset}{\mathbb}
\newcommand{\N}{\numberset{N}}
\newcommand{\Z}{\numberset{Z}}

\newcommand{\R}{\numberset{R}}

\newcommand{\F}{\numberset{F}}

\newcommand{\mV}{\mathcal{V}}

\newcommand{\mL}{\mathcal{L}}

\newcommand{\mI}{\mathcal{I}}

\newcommand{\wt}{\textnormal{wt}}

\newcommand{\disc}{\textnormal{disc}}

\newcommand{\mmod}[1]{\mbox{$\;(\textnormal{mod}\;{#1})$}}

\newcommand{\T}{^{\sf T}}
\newcommand{\GL}{\textnormal{GL}}

\newcommand{\subspace}[1]{\mbox{$\langle{#1}\rangle$}}
\newcommand{\type}{\textnormal{type}}

\newcommand{\NOne}{\mbox{\textnormal{N}$_1$}}
\newcommand{\NA}{\mbox{\textnormal{N}$_{0,\text{a}}$}}
\newcommand{\NNA}{\mbox{\textnormal{N}$_{0,\text{na}}$}}
\newcommand{\Binom}[2]{{\genfrac{(}{)}{0pt}{1}{#1}{#2}}}

\newcommand{\Gaussian}[2]{{\genfrac{[}{]}{0pt}{1}{#1}{#2}}}
\newcommand{\GaussianD}[2]{\genfrac{[}{]}{0pt}{0}{#1}{#2}}

\newcommand*{\myproofname}{Proof}

\usepackage{authblk}

\usepackage{array}
\newcolumntype{x}[1]{>{\centering\arraybackslash\hspace{0pt}}p{#1}}
\newcolumntype{y}[1]{>{\centering\arraybackslash\hspace{0pt}}m{#1}}

\renewcommand{\to}{\longrightarrow}

\usepackage{setspace}
\title{\textbf{$\ell$-Complementary Subspaces and Codes \\ in Finite Bilinear Spaces}}

\author[1]{Heide Gluesing-Luerssen\thanks{H. G-L was partially supported by the grant \#422479 from the Simons Foundation.}}
\author[2]{Alberto Ravagnani\thanks{A. R. was partially supported by the Dutch Research Council through the grants VI.Vidi.203.045,
OCENW.KLEIN.539,
and by the Royal Academy of Arts and Sciences of the Netherlands.}}

\affil[1]{Department of Mathematics, University of Kentucky, USA}
\affil[2]{Department of Mathematics and Computer Science, Eindhoven University of Technology, the Netherlands}

\date{December 8, 2022}

\usepackage{enumitem}
\setitemize{itemsep=0pt}
\setenumerate{itemsep=0pt}

\begin{document}
\maketitle
	
\begin{abstract}\label{sec:Abstract}
We consider (symmetric, non-degenerate) bilinear spaces over a finite field and investigate the properties of their
$\ell$-complementary subspaces, i.e., the subspaces that intersect their dual in dimension $\ell$.
This concept generalizes that~of a totally isotropic subspace and, in the context of coding theory, specializes to the notions of self-orthogonal, self-dual and linear-complementary-dual (LCD) codes.
In this paper, we focus on the
enumerative and asymptotic combinatorics of all these objects, giving formulas for their numbers and describing their typical behavior (rather than the behavior of a single object).
For example, we give a closed formula for the average weight distribution of an $\ell$-complementary code in the Hamming metric, generalizing a result by Pless and Sloane on the aggregate weight enumerator of binary self-dual codes.
Our results also show that self-orthogonal codes, despite being very sparse in the set of codes of the same dimension over a large field, asymptotically behave quite similarly to a typical, not necessarily self-orthogonal, code.
In particular, we prove that most self-orthogonal codes are MDS over a large field by computing the asymptotic proportion of the non-MDS ones for growing field size. 
\end{abstract}

\section{Introduction}

\textit{Self-orthogonal} (or \textit{totally isotropic} or sometimes just \textit{isotropic}) subspaces  over finite fields with respect to a symmetric, non-degenerate bilinear form are a classical topic in finite geometry, algebra, and group theory; see for instance~\cite{Taylor92,dieudonne1963geometrie,artin2016geometric} and the references therein. These are subspaces that are contained in their dual with respect to the underlying bilinear form. In the extreme case where their dimension is exactly half of that of the ambient space, self-orthogonal subspaces coincide with their dual and are thus called \textit{self-dual}.

When the ambient space is $\F_q^n$ and the bilinear form is the standard inner product,
self-orthogonal subspaces are called self-orthogonal \textit{codes}.
They have been extensively studied in coding theory over the past few decades, especially in connection with the invariant theory of error-correcting codes and their \textit{MacWilliams identities}; see~\cite{nebe2006self} for a general reference.
At the time of writing this paper, several questions about self-orthogonal and self-dual codes are still open, most notably the existence of a binary self-dual code of length 72 and minimum Hamming distance~16. In~\cite{CS96,S96}, self-orthogonal codes were used to construct the currently best known family of quantum codes based on classical codes.

In this paper, we consider the notion of an \textit{$\ell$-complementary} subspace over a finite field with respect to a symmetric, bilinear, non-degenerate form. These are the subspaces that intersect their dual in dimension exactly $\ell$ and thus generalize all the concepts we just discussed.
The intersection of a subspace with its dual is called the \textit{radical} of the subspace.
In the context of coding theory, the radical of a code is often called its \textit{hull}. It
plays an important role
in determining the complexity of certain algorithms that are relevant to code-based cryptography; see e.g.~\cite{sendrier2000finding}.
For $\ell=0$ we obtain the notion of a \textit{linear complementary dual} (\textit{LCD}) code,
whose study has been initiated by Massey~\cite{massey1992linear} in 1992.
More recently,
LCD codes
have received increasing attention in connection with information theory security~\cite{carlet2016complementary}.

In contrast with most of the literature on finite bilinear spaces, in this paper we do not focus on $\ell$-complementary subspaces as ``individual'' objects.
Instead, we investigate the \textit{average} and \textit{typical} properties
of the entire family of $\ell$-complementary spaces.
Furthermore, different from most of the coding theory literature, we often prove our results for an arbitrary finite bilinear space (i.e., we do not restrict to the standard inner product on $\F_q^n$).
As the reader will see, our results generally depend on the type of the underlying bilinear form, but the
differences almost become unnoticeable over sufficiently large fields.
In this kind of study, questions with a strong combinatorial flavor naturally arise. For example:

\begin{itemize}\setlength\itemsep{0em}
    \item[(Q1)] What is the asymptotic proportion of $\ell$-complementary subspaces over large fields?
     \item[(Q2)] What is the average value of the weight distribution of an $\ell$-complementary code?
     \item[(Q3)] Does the typical $\ell$-complementary code have good distance properties?
\end{itemize}

To answer these questions, we often use a hybrid approach that combines techniques from order theory with classical results in algebra and finite geometry.
We survey and sometimes re-establish the latter results using concise arguments.
This approach allows us
to capture how the algebra of the underlying field interacts with the combinatorics of $\ell$-complementary spaces.

We conclude the introduction by outlining the organization of the paper, which also gives us the chance to briefly illustrate its main contributions.

\paragraph*{Outline.}

In Section~\ref{sec:ellcompl}
we review and extend the theory of self-orthogonal subspaces of finite bilinear spaces, focusing on the
enumerative aspects.
We then introduce the concept of an $\ell$-complementary subspace and establish a formula for the number of such subspaces with a given dimension (for any bilinear, symmetric and nondegenerate form).
In Section~\ref{sec:3}
we focus on the standard inner product
of $\F_q^n$ as a bilinear form, which leads to the concept of
$\ell$-complementary error-correcting code. By specializing the previous results, we obtain closed formulas for the number of $\ell$-complementary codes of any given dimension.
In Section~\ref{sec:4} we study the average weight distribution of an $\ell$-complementary code endowed with the Hamming metric, giving a closed formula for it and generalizing a classical result by Pless and Sloane.
In Section~\ref{sec:5} we study the asymptotics of the average weight distribution of self-orthogonal codes, and use it to describe the overall behavior of this family of codes over large fields.
Finally, in Section~\ref{sec:6}
we prove that -- over a sufficiently large field --
MDS self-orthogonal codes are dense within the set of self-orthogonal codes of the same dimension. We also compute the asymptotic density of the non-MDS codes, and compare it with the asymptotic density of (not necessarily self-orthogonal) non-MDS codes.

\section{\texorpdfstring{$\ell$}{}-Complementary Subspaces over Finite Fields}
\label{sec:ellcompl}

Throughout this paper, $q$ denotes a prime power and $\F_q$ is the finite field with $q$ elements.
The main goal of this section is to give a general, closed formula for the number of $\ell$-complementary $k$-dimensional subspaces of an $\F_q$-linear space $V$ with respect to a given symmetric, bilinear and non-degenerate form~$B$.
By definition, a subspace is $\ell$-complementary if it intersects its orthogonal (with respect to~$B$) in dimension exactly~$\ell$.

We will take a hybrid approach that combines algebra and combinatorics. We start by establishing the notation for the rest of the paper and by recalling some known facts on bilinear spaces.

\begin{notation}\label{N-VB}
Throughout this paper $n \ge 1$ is an integer, $V$ is an $n$-dimensional $\F_q$-vector space, and~$B$ is a symmetric, bilinear, non-degenerate
form on $V$. The pair $(V,B)$ is called a \textbf{bilinear space}.
Unless otherwise stated, $k$ denotes an integer with $0 \le k \le n$.
\end{notation}

Recall that the form~$B$ is \textbf{alternating} if $B(v,v)=0$ for all $v\in V$.
Since~$B$ is also symmetric, this can only occur if~$q$ is even.
Moreover, in this case~$n$ must be even as well; see e.g.~\cite[p.~142]{Moor17}.

Let $\{e_1,\ldots,e_n\}$ be an ordered basis of~$V$. Then the matrix
\[
    G=\big(B(e_i,e_j)\big)_{i,j=1,\ldots,n}
\]
is called a \textbf{Gram matrix} of $(V,B)$.
It is non-singular and unique up to congruence $G\sim MGM\T$, where $M\in\GL_n(\F_q)$.
Define $Q_q$ as the set of nonzero squares in~$\F_q$.
Note that $Q_q$ is a subgroup of~$\F_q^*$ and $Q_q=\F_q^*$ if $q$ is even and $|Q_q|=(q-1)/2$ otherwise.
The determinants of different Gram matrices of~$B$ differ by a scalar factor in~$Q_q$.
All of this provides us with a well-defined invariant
\[
     \disc(V,B)=\det(G)Q_q\in\F_q^*/Q_q,
\]
which is called the \textbf{discriminant} of the bilinear space $(V,B)$.
We say that $\disc(V,B)$ \textbf{is a square} if $\det(G)\in Q_q$ for any (hence every) Gram matrix~$G$, and that it \textbf{is a nonsquare} otherwise.
Note that for even~$q$, the discriminant is always a square and therefore does not carry any information.

We define the following basic notions, which all depend on the choice of the form~$B$.
For ease of notation, we do not indicate this dependence.

\begin{definition}\label{def:Notation}
Let $C \le V$ be a $k$-dimensional subspace. The \textbf{orthogonal} of $C$ is the subspace
\[
     C^\perp=\{v \in V \mid B(v,w)=0 \mbox{ for all $w \in C$}\} \le V.
\]
The \textbf{radical} of $C$ is $C\cap C^\perp$.
If $\ell\in\{0,\ldots,k\}$, then $C$ is called an $\ell$-\textbf{complementary} subspace if $\dim(C\cap C^\perp)=\ell$.
We say that $C$ is \textbf{self-orthogonal} if it is $k$-complementary, that is $C \subseteq C^\perp$, and that~$C$ is \textbf{LCD} if
it is $0$-complementary (linear-complementary-dual).
If $2k=n$ and $C$ is self-orthogonal, then $C$ is \textbf{self-dual}. Finally, a vector $v \in V$ is called \textbf{self-orthogonal} if the subspace $\subspace{v} \le V$ is self-orthogonal.
\end{definition}

In the coding literature, the radical of a subspace~$C$ of~$\F^n$, i.e., $C\cap C^\perp$, is often called its \textbf{hull}; see~\cite{sendrier1997dimension}.
Since this space is the radical of the bilinear form restricted to~$C$, we prefer to continue calling it the radical.
Self-orthogonal subspaces are also called \textbf{totally isotropic} and self-orthogonal vectors are also known as \textbf{isotropic} or \textbf{singular} vectors.

In this section we are primarily concerned with the following sets and quantities.

\begin{notation} \label{not:qties}
For $k\in\{0,\ldots,n\}$ and $\ell\in\{0,\ldots,k\}$, let
\[
   \Sigma(V,B,k,\ell)=\{C\leq V\mid \dim(C)=k \mbox{ and }  C\text{ is $\ell$-complementary}\}.
\]
We also define $\sigma(V,B,k,\ell):=|\Sigma(v,B,k,\ell)|$ and $\sigma(V,B,k):=\sigma(V,B,k,k)$, which is the number of
$k$-dimensional self-orthogonal subspaces of $V$.
\end{notation}


The non-degeneracy of~$B$ implies the following standard consequences. The proof is omitted.

\begin{lemma}
Let $C,D \le V$. Then $\dim(C^\perp) = n-\dim(C)$,
$(C+D)^\perp=C^\perp \cap D^\perp$, and
$(C \cap D)^\perp = C^\perp + D^\perp$.
In particular,
an LCD subspace~$C \le V$ satisfies $C\oplus C^\perp=V$. Furthermore, $\sigma(V,B,k,\ell)=0$ for $\ell>\min\{k,n-k\}$.
\end{lemma}

As already mentioned, the goal of this section is to determine the numbers $\sigma(V,B,k,\ell)$ explicitly.
They clearly depend on the choice of the bilinear form~$B$.
As it turns out, they are fully determined by the maximum possible dimension of a self-orthogonal subspace of $(V,B)$.

We will start with determining $\sigma(V,B,k)$.
These numbers are in fact known from~\cite{Pless65}, which is based on~\cite{segre1959geometrie}.
Since~\cite{Pless65} is hard to obtain, we will derive these results again and in a quite different fashion.
We will do so in a recursive manner by using the numbers $\sigma(V,B,1)$, which can be found in more easily accessible
literature on finite geometry.

The following fact is a consequence of  Witt's Theorem, which pertains to the extendability of isometries on subspaces of~$V$
to isometries on the full space~$V$.

\begin{theorem}[\text{see~\cite[p.~59]{Taylor92}}]
\label{thwitt}
All maximal self-orthogonal subspaces of~$V$ have the same dimension. As a consequence, every self-orthogonal subspace $C \le V$ can be extended to a maximal one.
\end{theorem}

\begin{definition}
The common dimension defined in Theorem~\ref{thwitt}
is called the the
\textbf{Witt index} of the bilinear space $(V,B)$ and denoted by $w(V,B)$.
\end{definition}

By definition, $\sigma(V,B,k)=0$ for all $k>w(V,B)$ and $\sigma(V,B,k,\ell)=0$ for $\ell>w(V,B)$ because $C\cap C^\perp$ is self-orthogonal.

The Witt index of a bilinear space takes the following values.

\begin{theorem}[\mbox{\cite[pp.~150]{Moor17}} and \mbox{\cite[Cor.~3.1]{Pless64}} for $q$ even]\ \label{thm:witt}
\begin{alphalist}
\item Let~$n$ be odd. Then $w(V,B)=(n-1)/2$.
\item Let $n$ be even. If~$q$ is even we have $w(V,B)=n/2$ and if $q$ is odd we have
         \[
             w(V,B)=\left\{\begin{array}{cl}n/2& \text{if }(-1)^{n/2}\,\disc(V,B)\in Q_q,\\[.5ex]n/2-1&\text{otherwise.}\end{array}\right.
         \]
\end{alphalist}
\end{theorem}

For the remainder of this paper
it will be very convenient to classify bilinear spaces according to their ``type'', defined as follows.

\begin{definition}\label{def:types}
We introduce the following \textbf{types}.
\begin{alphalist}
\item Let~$q$ be odd. The space $(V,B)$ is said to be of type \textbf{(P)} if~$n$ is odd; it is of type \textbf{(H)} if $n$ is even and
$w(V,B)=n/2$, and it is of type \textbf{(E)} if $n$ is even and $w(V,B)=n/2-1$. The types stand for ``parabolic'', ``hyperbolic'', and ``elliptic'', respectively.
\item Let~$q$ be even. We say that $(V,B)$ is of
type \textbf{(N$_{\textbf{1}}$)} if $n$ is odd; it is of type \textbf{(N$_{\textbf{0,a}}$)} if~$n$ is even and $B$ is alternating, and
of type \textbf{(N$_{\textbf{0,na}}$)} if~$B$ is non-alternating.
\end{alphalist}
\end{definition}

Note that if~$q$ is odd, then~$n$ together with $w(V,B)$ fully determine the type, whereas for even~$q$ the Witt index does not distinguish between (\NA) and (\NNA).

We now turn to the problem of counting the number of self-orthogonal subspaces of a given dimension for each of the above types.
The following result will be a main tool in our approach.
A similar result appears in \cite[Thm.~24.3]{Moor17}.

\begin{proposition}\label{prop:quotient}
Let $w=w(V,B)$ and let $U\leq V$ be a self-orthogonal subspace of dimension  $0\leq r\leq w$.  On the quotient space
$U^\perp/U$ define the bilinear form $B_U$ via $B_U(v+U,w+U)=B(v,w)$.
Then $B_U$ is  well-defined and a symmetric, bilinear, non-degenerate form.
Moreover, for every $k\in\{r,\ldots,w\}$ the projection $\pi: U^\perp \to U^\perp/U$ induces a bijection between the sets
\[
    \{C\leq V\mid \dim C=k,\, U\leq C,\, C\text{ self-orthogonal w.r.t. }B\}
\]
and
\[
    \{\tilde{C}\leq U^\perp/U\mid \dim\tilde{C}=k-r,\, \tilde{C}\text{ self-orthogonal w.r.t. }B_U\}.
\]
In particular, $w(U^\perp/U,B_U)=w(V,B)-r$.
Furthermore, if the type of $(V,B)$ is not (\NNA), then
$(U^\perp/U,B_U)$ is of the same type as $(V,B)$ unless $r=n/2$, in which case $U^\perp/U=\{0\}$.
\end{proposition}

The last statement about the invariance of the type is not true for type (\NNA), because~$B_U$ may be alternating even if~$B$ is not.
We will return to this later in \cref{ex:InducedAlternating}  and \cref{prop:SOContainx}, where we determine the number of self-orthogonal codes in $\F_q^n$ containing a given nonzero vector.
This will then be used in \cref{cor:sumweight} to determine the average Hamming weight distribution of self-orthogonal codes.

\begin{proof}[Proof of Proposition~\ref{prop:quotient}]
The well-definedness and non-degeneracy of $B_U$ is easily verified.
For the rest note that the containment $U\leq C$ for a  self-orthogonal space~$C$ implies $C\leq C^\perp\leq U^\perp$.
Hence the subspaces in the first set of the statement are actually contained in $U^\perp$.
Moreover, one easily checks that $\pi(C)^\perp=\pi(C^\perp)$, and thus the bijection follows.
The conclusion about the Witt index follows now directly from \cref{thwitt}.
Finally, the invariance of the type for odd~$q$ follows from the fact that $\dim(U^\perp/U)=n-2r$, which has the same parity as~$n$,
together with $n/2-r=(n-2r)/2$, which lets us infer the type from the Witt index.
For even~$q$ and odd~$n$, the invariance of the type is straightforward.
Furthermore, it is clear that the form~$B_U$ is alternating, if~$B$ is.
\end{proof}

\begin{notation} \label{not:B_U}
In the remainder of the paper, for a self-orthogonal space $U \le V$ we let $B_U$ denote the bilinear form defined in Proposition~\ref{prop:quotient}.
\end{notation}

Proposition~\ref{prop:quotient} has the following immediate consequence.

\begin{corollary}\label{cor:SOCodesContaining}
Let  $U\leq V$ be a self-orthogonal subspace of dimension $0\le r \leq w(V,B)$.
The number of $k$-dimensional self-orthogonal subspaces $C \le V$ containing $U$ is
$\sigma(U^\perp/U,B_U,k-r)$.
\end{corollary}

The above result allows us to determine, for most types, the number of $k$-dimensional self-orthogonal codes recursively. We just need the number of $1$-dimensional self-orthogonal codes.
Note that the case (\NA) below is clear, and so is the vector space property of~$\mI$.

\begin{theorem}[\mbox{see \cite[p.~153]{Moor17} for (P), (H), (E) and \cite[Cor.~2.2]{Pless64} for (\NOne) and (\NNA)}]\label{thm:SOVectors}\ \\
Let $w=w(V,B)$. Then
\[
     \sigma(V,B,1)=\left\{\begin{array}{cl}
         \frac{(q^w-1)(q^{n-w-1}+1)}{q-1}&\text{for type (P), (H), (E)},\\[1ex]
         \frac{q^n-1}{q-1}&\text{for type (\NA)},\\[1ex] \frac{q^{n-1}-1}{q-1}&\text{for type (\NOne) and (\NNA).}\end{array}\right.
\]
Furthermore, if $q$ is even then the set $\mI=\{v\in V\mid B(v,v)=0\}$ is a subspace of~$V$.
We have $\mI=V$ if~$B$ is alternating, and
$\dim\mI=n-1$
if~$B$ is non-alternating.
\end{theorem}

Combining the previous theorem
with Proposition~\ref{prop:quotient} we obtain a recursive formula for the number of $k$-dimensional self-orthogonal subspaces.
This works for all types except (\NNA).

\begin{theorem} \label{thm:witt2}
Suppose $\dim V=n\ge2$ and let $w=w(V,B)$. Suppose the type of $(V,B)$ is not~(\NNA).
Let $1 \le k \le w$.
Then
\[
     \sigma(V,B,k)\!=\!\sigma(V,B,k\!-\!1) Q,\; \text{ with }
      Q=\!\left\{\!\!\begin{array}{cl}
       \frac{(q^{w-k+1}-1)(q^{n-k-w}+1)}{q^k-1}&\!\!\text{if the type is (P), (H), or (E)},\\[1ex]
       \frac{(q^{n-2k+2}-1)}{q^k-1}&\!\!\text{if the type is (\NA)},\\[1ex]
        \frac{(q^{n-2k+1}-1)}{q^k-1}&\!\!\text{if the type is (\NOne)}.\\[1ex]
       \end{array}\right.
\]
\end{theorem}

\begin{proof}
Consider the set
\begin{equation}\label{e-setS}
    S=\big\{(W,W') \,\big|\, W \le W' \le V, \,  \mbox{$W'$ is self-orthogonal}, \, \dim(W)=k-1, \,
    \dim(W') =k\big\}.
\end{equation}
Note that every subspace~$W$ of a self-orthogonal space~$W'$ is also self-orthogonal.
We count the elements of~$S$ in two ways.
On the one hand we have
\begin{equation}\label{e-setS2}
   |S|= \GaussianD{k}{k-1}_q \cdot \big|\{W' \le V\mid  W' \mbox{ is self-orthogonal}, \, \dim(W')=k\}\big|
        =\GaussianD{k}{1}_q\sigma(V,B,k).
\end{equation}
On the other hand, thanks to \cref{prop:quotient}, the type of $\smash{(W^{\perp}/W, B_W)}$ agrees with the type of $(V,B)$ and thus
is the same for all subspaces~$W$ of dimension $k-1$.
Hence we obtain
\begin{align*}
    |S|&= \sum_{\substack{W \le V \\ W \textnormal{ self-orthogonal} \\ \dim(W)=k-1}}\hspace*{-.5em}
                \big|\{L \le W^{\perp}/W \mid L \mbox{ is self-orthogonal w.r.t. } B_W, \, \dim(L)=1\}\big|\\
         &=\sigma(V,B,k-1)\, \sigma(W^{\perp}/W,B_W,1).
\end{align*}
Thus we have
\[
  \sigma(V,B,k)=\frac{q-1}{q^k-1}\, \sigma(V,B,k-1)\, \sigma(W^{\perp}/W,B_W,1).
\]
By \cref{prop:quotient},  the Witt index of $(W^{\perp}/W, B_W)$ is $w-k+1$ and $\dim(W^{\perp}/W, B_W)=n-2(k-1)$.
Now \cref{thm:SOVectors} leads to the stated results.
\end{proof}

Now we are ready to present $\sigma(V,B,k)$ explicitly for all types except (\NNA).
For the remaining type we cite the formula given in \cite{Pless65}.
The second expressions for $\sigma(V,B,k)$ in (a) can also be found in \cite{Pless65}.

\begin{theorem}[\mbox{\cite[p.~421]{Pless65} for (c)}]\label{thm:Sigma}
Let $\dim V=n\ge2$ and $w=w(V,B)$. 
Let $1 \le k \le w$. Set $a_k=\prod_{i=1}^k(q^i-1)$.
\begin{alphalist}
\item Let the type of $(V,B)$ be (P), (H), (E), or (\NOne). Then
        \begin{align*}
                \sigma(V,B,k) &= \frac{\prod_{i=1}^k (q^{n-i-w}+1)(q^{w-i+1}-1)}{a_k}\\[2ex]
                &=\left\{\begin{array}{cl}
     \frac{\prod_{i=1}^{k}(q^{n+1-2i}-1)}{a_k}
     &\text{for type (P) and (\NOne)},\\[2ex]
     \frac{(q^{n-k}+q^{n/2}-q^{n/2-k}-1)\prod_{i=1}^{k-1}(q^{n-2i}-1)}{a_k}
     &\text{for type (H)}, \\[2ex]
     \frac{(q^{n-k}-q^{n/2}+q^{n/2-k}-1)\prod_{i=1}^{k-1}(q^{n-2i}-1)}{a_k}
     &\text{for type (E)}.
     \end{array} \right.
        \end{align*}
\item For type~(\NA) we have
         \[
                \sigma(V,B,k) 
                = \frac{\prod_{i=1}^{k}(q^{n-2i+2}-1)}{a_k}.
        \]
\item Let the type be (\NNA). Then
         \[
   \sigma(V,B,k)=\frac{(q^{n-k}-1) \prod_{i=1}^{k-1}(q^{n-2i}-1)}{a_k}.
        \]
\end{alphalist}
\end{theorem}

\begin{proof}
Parts (a) and (b): For $k=1$ the formulas agree with \cref{thm:SOVectors}. The rest follows from induction using \cref{thm:witt2} and the fact that $\sigma(V,B,0)=1$ for any type.
For part (c) we refer to~\cite{Pless65}.
\end{proof}

The only case we cannot derive with our recursion in \cref{thm:witt2} is the type (\NNA), because if $B$ is non-alternating then it depends on
the subspace~$W$ whether $B_W$ in \cref{prop:quotient} is alternating or not; see also Proposition~\ref{ex:InducedAlternating} in the next section.
However, we can use the above expression for $\sigma(V,B,k)$ in type (\NNA) to determine the number of $k$-dimensional
subspaces~$W$ such that $B_W$ is alternating.

\begin{proposition}\label{prop:NonAltInduced}
Let $(V,B)$ be of type (\NNA) and let $1\leq k\leq w(V,B)=n/2$.
Then
\[
       \big|\{W\leq V\mid \dim W=k,\,W\text{ self-orthogonal and $B_W$ is alternating}\}\big|=\prod_{i=1}^{k-1}\frac{q^{n-2i}-1}{q^i-1}.
\]
\end{proposition}

\begin{proof}
Consider again the set~$S$ from \eqref{e-setS}. Then as in \eqref{e-setS2} we have $\smash{|S|=\Gaussian{k}{1}_q \,\sigma(V,B,k)}$.
On the other hand, let $\smash{\alpha:= \big|\{W\leq V\mid \mbox{$\dim W=k-1$},\,W \text{ self-orthogonal and $B_W$ is alternating}\}\big|}$. Then  \cref{thm:SOVectors} implies
\begin{align*}
     |S|&=\sum_{\substack{W \le V \\ W \textnormal{ self-orthogonal} \\ \dim(W)=k-1}}\hspace*{-.5em}
                \big|\{L \le W^{\perp}/W \mid L \mbox{ is self-orthogonal w.r.t. } B_W, \, \dim(L)=1\}\big|\\[1ex]
     &=\alpha\,\frac{q^{n-2(k-1)}-1}{q-1}+\big(\sigma(V,B,k-1)-\alpha\big)\frac{q^{n-2(k-1)-1}-1}{q-1}\\
     &=\alpha q^{n-2k+1}+\sigma(V,B,k-1)\frac{q^{n-2k+1}-1}{q-1}.
\end{align*}
Equating this with the above expression for $|S|$ and using \cref{thm:Sigma}(c) leads to
\begin{align*}
    \alpha&=\frac{(q^k-1)\sigma(V,B,k)-(q^{n-2k+1}-1)\sigma(V,B,k-1)}{q^{n-2k+1}(q-1)}\\[1ex]
        &=\frac{\prod_{i=1}^{k-2}(q^{n-2i}-1)}{q^{n-2k+1}(q-1)\prod_{i=1}^{k-1}(q^i-1)}
              \Big((q^{n-k}-1)(q^{n-2k+2}-1)-(q^{n-2k+1}-1)(q^{n-k+1}-1)\Big)\\[1ex]
        &=\frac{\prod_{i=1}^{k-2}(q^{n-2i}-1)}{q^{n-2k+1}(q-1)\prod_{i=1}^{k-1}(q^i-1)}q^{n-2k+1}(q^k+1-q^{k-1}-q)\\[1ex]
        &=\prod_{i=1}^{k-2}\frac{q^{n-2i}-1}{q^i-1}.
\end{align*}
The statement follows.
\end{proof}

We can finally return to the main goal of this section, which is
determining the number of $\ell$-complementary subspaces of $(V,B)$ for general~$\ell$.
Our approach uses M\"obius inversion in the lattice of subspaces~\cite{stanley2011enumerative}.
This allows us to reduce the computation to the case of self-orthogonal codes, which we discussed earlier in this section. Note that the following result generalizes~\cite[Theorem 4.5]{sendrier1997dimension} to an arbitrary bilinear form.
Recall the notation $\Sigma(V,B,k,\ell)$ from \cref{def:Notation}.

\begin{theorem}\label{thm:sigmakell}
Let $k\in\{0,\ldots,n\}$ and $\ell\in\{0,\ldots,k\}$. Then
\[
   \sigma(V,B,k,\ell)=\sum_{s=\ell}^{w(V,B)}\sigma(V,B,s)\GaussianD{s}{\ell}_q\GaussianD{n-2s}{k-s}_q
   (-1)^{s-\ell}\,q^{\Binom{s-\ell}{2}}.
\]
\end{theorem}

\begin{proof}
Let $\mL$ denote the lattice of subspaces of $V$, ordered by inclusion. Define functions $f,g:\mL \to \Z$ by
\begin{align*}
    f(U) &=\big|\{C \le V \mid \dim(C)=k, \, C \cap C^\perp=U \}\big|, \\[.7ex]
    g(U) &= \sum_{\substack{W \in \mL \\ W \ge U}} f(W).
\end{align*}
By definition,
$g(U)$ counts the number of $k$-dimensional subspaces $C$ such that $U \le C \le U^\perp$. Therefore,
for all $U \in \mL$ of dimension $s$ we have
\[
g(U) = \left\{\begin{array}{cl}
 0 \ \  & \mbox{if $U$ is not self-orthogonal,} \\[.7ex]
\GaussianD{n-2s}{k-s}_q & \mbox{if $U$ is self-orthogonal.}
\end{array}\right.
\]
By applying M\"obius inversion in the lattice~$\mL$ we obtain, for all $U \in \mL$ of dimension $\ell$,
\[
f(U)=\sum_{s=0}^n\sum_{\substack{W\in\Sigma(V,B,s,s)\\ U\leq W}}
\GaussianD{n-2s}{k-s}_q \, (-1)^{s-\ell} \, q^{\binom{s-\ell}{2}}.
\]
Now $\sigma(V,B,k,\ell)$ is obtained as follows:
\[
  \sigma(V,B,k,\ell)=\sum_{\substack{U \in \mL \\ \dim(U)=\ell}} f(U)= \, \sum_{s=0}^n \; \sum_{W\in\Sigma(V,B,s,s)} \GaussianD{s}{\ell}_q \GaussianD{n-2s}{k-s}_q \, (-1)^{s-\ell} \, q^{\binom{s-\ell}{2}}.
\]
Since the cardinality of $\Sigma(V,B,s,s)$ is $\sigma(V,B,s)$, the desired result follows.
\end{proof}

We conclude this section by observing that Theorem~\ref{thm:sigmakell} allows us to compute the aggregate dimension of the radical of a $k$-dimensional subspace of $(V,B)$.
The result generalizes~\cite[Corollary~4.20]{sendrier1997dimension}, where the asymptotic version of the following is given for the special
bilinear space $(\F_q^n,\,\cdot\,)$.
Of course, the average dimension of the radical of a $k$-dimensional subspace of $(V,B)$ can be computed from Corollary~\ref{cor:cumulrad} by simply dividing by the $q$-ary coefficient $\smash{\Gaussian{n}{k}_q}$.

\begin{corollary} \label{cor:cumulrad}
Let $k \in \{0, \ldots, n\}$. We have
\[
\sum_{\substack{C \le V \\ \dim(C)=k}} \dim(C \cap C^\perp) = \sum_{\ell=1}^{k} \ell \sum_{s=\ell}^{w(V,B)}\sigma(V,B,s)\GaussianD{s}{\ell}_q\GaussianD{n-2s}{k-s}_q
   (-1)^{s-\ell}\,q^{\Binom{s-\ell}{2}}.
\]
\end{corollary}

\begin{proof}
The sum on the LHS is $\smash{\sum_{\ell=0}^{k} \ell \, \sigma(V,B,k,\ell)}$.
We therefore conclude by Theorem~\ref{thm:sigmakell}.
\end{proof}

\section{\texorpdfstring{$\ell$}{}-Complementary Codes}
\label{sec:3}

In this section we focus on the case where $V=\F_q^n$ and $B$ is the standard inner product (also known as the \textit{dot product}).
The bilinear space $(\F_q^n, \cdot)$ is the ambient space of classical coding theory; see~\cite{macwilliams1977theory} for a general reference. In this context, the subspaces of $\F_q^n$ are called (\textbf{error-correcting}) \textbf{codes}.

We describe the type of the standard inner product and its Witt index in terms of the values of $q$ and $n$.
These
results are interesting in their own right but will also be needed later in the paper.
We then specialize Theorem~\ref{thm:Sigma}
to obtain the number of self-orthogonal error-correcting codes. The result can also be found in~\cite{sendrier1997dimension}.

\begin{notation}\label{nota:codes}
A $k$-dimensional code $C \le \F_q^n$ is called an \textbf{$[n,k]_q$-code}. For convenience, we set
$\Sigma_q(n,k,\ell)=\Sigma(\F_q^n,\cdot,k,\ell)$ and
$\sigma_q(n,k,\ell)=|\Sigma(\F_q^n,\cdot,k,\ell)|$.
Furthermore, we define
$\Sigma_q(n,k)=\Sigma_q(n,k,k)$ and $\sigma_q(n,k)=\sigma_q(n,k,k)$.
We denote by $\wt(v)$ the (\textbf{Hamming}) \textbf{weight} (i.e., the number of nonzero components)
of a vector $v \in \F_q^n$. The \textbf{minimum} (\textbf{Hamming}) \textbf{distance} of a nonzero code $C \le \F_q^n$ is
$d(C)=\min\{\wt(v) \mid v \in C, \, v \neq 0\}$.
\end{notation}

We start by listing the Witt index of the bilinear space $(\F_q^n,\cdot)$, i.e., the maximal dimensions of self-orthogonal codes in $\F_q^n$.
The result can also be found in, e.g., \cite{Pless65}.
For further reference we introduce the type of a pair $(q,n)$.
It simply translates \cref{def:types} and \cref{thm:witt} into conditions on~$q,n$.
Note that type (\NA) does not arise for $(\F_q^n,\,\cdot\,)$.

\begin{proposition}\label{prop:typeqn}
We define $\type(q,n)$ as the type of the space $(\F_q^n,\,\cdot\,)$, that is,
\[
   \type(q,n):=\left\{\begin{array}{cl}\text{(\NNA)} &\text{if $n,\,q$ even},\\
         \text{(\NOne)} &\text{if $n$ odd and $q$ even},\\
         \text{(P)} &\text{if $n,\,q$  odd},\\
        \text{(H)} &\text{if $q$ odd and $n$ even and }
        [n\equiv0 \mmod{4} \text{ or }q\equiv1\mmod{4}],\\
        \text{(E)} &\text{if }
         n\equiv2 \mmod{4} \text{ and }q\equiv3\mmod{4}.
         \end{array}\right.
\]
Now we have
\[
   w(\F_q^n,\cdot)=
        \left\{\begin{array}{cll}
        (n-1)/2 & \text{if } \type(q,n)\in\{\text{(\NOne), (P)}\}, \\[.5ex]
        n/2 & \text{if } \type(q,n)\in\{\text{(\NNA), (H)}\},\\[.5ex]
        n/2-1 & \text{if } \type(q,n)\in\{\text{(E)}\}.
        \end{array}\right.
\]
\end{proposition}

\begin{proof}
The standard basis of $\F_q^n$ gives the identity matrix as Gram matrix.
Clearly, the dot product is never alternating (choose a standard basis vector), and thus type (\NA) does not occur.
Since $(-1)^r$ is a square in $\F_q$ if and only if $q$ is even or $r$ is even or $(q-1)\in 4\Z$, the statements follow from Theorem~\ref{thm:witt}.
\end{proof}

For the standard inner product $B$, we can characterize the spaces $U \le \F_q^n$ for which the form~$B_U$ is alternating; see Notation~\ref{not:B_U}.

\begin{proposition}\label{ex:InducedAlternating}
Let~$q$ and~$n$ be even and let $B$ be the standard inner product on $\F_q^n$.
Let~$U$ be a self-orthogonal subspace of $\F_q^n$.
Then
\[
     B_U\text{ is alternating}\Longleftrightarrow
     \textbf{1}\in U,
\]
where $\textbf{1}$ is the all-one vector in $\F_q^n$.
\end{proposition}

\begin{proof}
Note first that $\textbf{1}$ is a self-orthogonal vector, since~$n$ is even.
Next, let $v=(v_1,\ldots,v_n)\in U^\perp$. Then
$v\cdot v=0\Longleftrightarrow\sum_{i=1}^n v_i^2=0\Longleftrightarrow (\sum_{i=1}^n v_i)^2=0\Longleftrightarrow \sum_{i=1}^n v_i=0\Longleftrightarrow v\cdot\textbf{1}=0$.
This establishes the stated equivalence.
\end{proof}

We conclude this section by specializing \cref{thm:Sigma} to the bilinear space~$(\F_q^n,\,\cdot\,)$.
This will be convenient later on.
The same formulas can be found in \cite[Prop.~3.2]{sendrier1997dimension}.

\begin{proposition}\label{P-sigmaqnk}
Let $n\geq2$ and $1\leq k\leq w(\F_q^n,\,\cdot\,)$ and set $a_{k}=\prod_{i=1}^k(q^i-1)$ and $b_{n,k}=\prod_{i=1}^{k-1}(q^{n-2i}-1)$.
Then
\[
   \sigma_q(n,k)=
   \left\{\begin{array}{cl}
   \displaystyle\frac{\prod_{i=1}^{k}(q^{n+1-2i}-1)}{a_{k}} &\text{if }  \type(q,n)\in\{\text{(\NOne), (P)}\},\\[3.5ex]
   \displaystyle\frac{(q^{n-k}-1)\,b_{n,k}}{a_{k}}&\text{if }\type(q,n)=\text{(\NNA)},\\[3.5ex]
   \displaystyle\frac{(q^{n-k}+\varepsilon q^{n/2}-\varepsilon q^{n/2-k}-1)\,b_{n,k}}{a_{k}}&\text{if }\type(q,n)\in\{\text{(H), (E)}\},
\end{array}\right.
\]
where $\varepsilon=1$ for type (H) and $\varepsilon=-1$ for type (E).
Moreover, $\sigma_{q}(n,0)=1$ for all $n\geq 1$.
\end{proposition}

\begin{remark} \label{rmk:gives_f}
As in \cref{thm:sigmakell}, Proposition~\ref{P-sigmaqnk} gives closed formulas for the number of $\ell$-complementary codes in $(\F_q^n,\,\cdot\,)$ for all
$\ell$.
They can also be found in \cite[Thm.~4.5]{sendrier1997dimension}.
\end{remark}

\begin{remark}
In \cite{carlet2018new},
the number of LCD codes in~$(\F_q^n,\,\cdot\,)$ (that is, $\ell=0$) was determined for the case where $q=2$ or $q$ is an odd prime power.
These results have been obtained by completely different methods, specifically by the action of the orthogonal group on the set of LCD codes, which leads to rather different expressions.
\end{remark}


\section{Weight Distributions}
\label{sec:4}

In this section we focus on the (Hamming) weight distribution of $\ell$-complementary codes. In contrast with most of the existing literature, which focuses on individual codes,
we look at the average behavior of
the weight distribution within the class of codes that share a common dimension
and are $\ell$-complementary.
Our main result is Theorem~\ref{T-AvgWeight}, which gives a closed formula for the mentioned
average weight distribution.
The proof relies on intermediate results that show how algebraic and combinatorial aspects of $\ell$-complementary codes interact with each other.
We follow the notation of Section~\ref{sec:3}.

\begin{definition}\label{def:Hamming}
The (\textbf{Hamming}) \textbf{support} of a vector $v \in \F_q^n$ is
defined as $\sigma(v)=\{i \mid v_i \neq 0\}$.
For a code $C \le \F_q^n$ and an integer $0 \le i \le n$,
we let~$A^{(i)}(C)$ denote the number of vectors $v \in C$ with $\wt(v)=i$.
We call $(A^{(i)}(C) \mid 0 \le i \le n)$ the \textbf{weight distribution} of $C$.
\end{definition}

In this section, we wish to compute the following quantities.

\begin{definition}\label{def:AvgDistr}
For $i \in \{0, \ldots, n\}$ define
\[
   A^{(i)}_q(n,k,\ell)=\displaystyle\sum_{C\in\Sigma_q(n,k,\ell)}\!\!\! A^{(i)}(C)
   \quad \text{ and }\
   \overline{A}_q^{(i)}(n,k,\ell) :=
   \frac{A^{(i)}_q(n,k,\ell)}{\sigma_q(n,k,\ell)}.
\]
We call $\smash{\big(A_q^{(i)}(n,k,\ell) \mid 0 \le i \le n\big)}$ the \textbf{aggregate weight distribution} of the set of all $\ell$-complementary $[n,k]_q$-codes and
$\smash{\big(\overline{A}_q^{(i)}(n,k,\ell) \mid 0 \le i \le n\big)}$ the \textbf{average weight distribution} of an $\ell$-complementary $[n,k]_q$-code.
\end{definition}

The aggregate weight distribution was computed by Pless and Sloane for the very special case of binary self-dual codes; see~\cite[Section~4]{pless1975classification}.

Note that $\smash{A_q^{(0)}(n,k,\ell)=\sigma_q(n,k,\ell)}$, and that $\smash{A_q^{(1)}(n,k,k)=0}$.
Furthermore, we clearly have $\smash{\sum_{i=0}^n\overline{A}_q^{(i)}(n,k,\ell)=q^k}$.
Thanks to the previous section we can compute the numbers $\sigma_q(n,k,\ell)$, and thus it suffices to compute
$\smash{A_q^{(i)}(n,k,\ell)}$.
We start by observing that, for all $i$,
\begin{equation} \label{note}
     A_q^{(i)}(n,k,\ell) =
     \sum_{\substack{v \in \F_q^n \\ \wt(v)=i}} |\{C \in\Sigma_q(n,k,\ell)\mid v \in C\}|.
\end{equation}
Again, we will begin with the self-orthogonal case, thus $\ell=k$.
Since any vector~$v\in\F_q^n$ contained in a self-orthogonal code is self-orthogonal, we will first count these vectors for any given weight.

\begin{notation}\label{nota:zeta}
For $i \in \{1, \ldots, n\}$ we let $\zeta_q(n,i) = |\{v \in \F_q^n \mid \wt(v)=i, \, v\cdot v =0\}|$.
\end{notation}

Using M\"obius inversion we can derive an explicit formula for $\zeta_q(n,i)$. Note that the following result gives $\zeta_q(n,1)=0$, as expected.

\begin{proposition}\label{prop:zeta}
We have $\zeta_q(n,0)=1$ and
\[
 \zeta_q(n,i)=\binom{n}{i}\bigg( (-1)^i+
    \sum_{j=1}^i\binom{i}{j}(-1)^{i-j}\Big((q-1)\sigma_q(j,1)+1\Big)\bigg)
  \ \text{ for }i\in \{1, \ldots, n\}.
\]
\end{proposition}

\begin{proof}
For a subset $S \subseteq \{1 \ldots n\}$ define
$f(S)$ as the number of self-orthogonal vectors $v \in \F_q^n$ whose support is equal to $S$. Then
$\smash{g(S):= \sum_{T\subseteq S} f(T)}$
is the number of self-orthogonal vectors $v \in \F_q^n$ whose support is contained in $S$.
Using that $\sigma_q(j,1)$ is the number of $1$-dimensional self-orthogonal subspaces in $\F_q^j$, we conclude that
\[
 g(S)=(q-1)\sigma_q(j,1) +1\text{ if }|S|=j>0\ \text{ and }\ g(\emptyset)=1.
\]
Now we use M\"obius inversion in the Boolean algebra over the set $\{1,\ldots,n\}$ and obtain for all sets $S$ of size $i$
\begin{align*}
  f(S)&= \sum_{T\subseteq S}(-1)^{i-|T|}g(T)
  =(-1)^ig(\emptyset)+\sum_{j=1}^i\sum_{\substack{T\subseteq S\\ |T|=j}}
    (-1)^{i-j}\Big((q-1)\sigma_q(j,1)+1\Big)\\
    &=(-1)^i+\sum_{j=1}^i\binom{i}{j}(-1)^{i-j}\Big((q-1)\sigma_q(j,1)+1\Big).
\end{align*}
Now the stated formula for $\zeta_q(n,i)$ follows.
\end{proof}

We next count the number of self-orthogonal codes containing a given self-orthogonal subspace. This quantity depends on $\type(q,n)$ and only in the non-alternating case also on the subspace.
It will be convenient to define the following quantity.
It should be noted that~$w$ is a separate parameter that is not necessarily equal to the Witt index of $(\F_q^n,\,\cdot\,)$. This separation will be necessary because of the change of type for certain induced spaces.

\begin{notation}\label{nota:tau}
Let $w\in\{1,\ldots,n/2\}$ and $k\in\{1,\ldots,w\}$. Define
\begin{equation}\label{e-tau}
   \tau_q(n,k,w)=\prod_{i=1}^{k}\frac{(q^{n-w-i}+1)(q^{w-i+1}-1)}{q^i-1}.
\end{equation}
Moreover, we set $\tau_q(n,0,w)=1$ for all $w\in\{0,\ldots,n/2\}$. Note that if $w=w(\F_q^n,\,\cdot\,)$ and $\type(q,n)\in\{\text{(P), (H), (E), (\NOne)}\}$ then $\tau_q(n,k,w)=\sigma_q(n,k)$; see \cref{thm:Sigma}(a).
\end{notation}

\begin{proposition}\label{prop:SOContainx}
Let $w=w(\F_q^n,\,\cdot\,)$ and $0\leq t\leq k\leq w$. Let $U\leq\F_q^n$ be a $t$-dimensional self-orthogonal subspace and set $\tilde{\sigma}_{k,U}=|\{C\in\Sigma_q(n,k)\mid U\leq C\}|$.
Then
\[
   \tilde{\sigma}_{k,U}=\left\{\begin{array}{ll}
       \tau_q(n-2t,k-t,w-t)&\text{if }\type(q,n)\in\{\text{(P), (H), (E), (\NOne)}\},\\[.5ex]
       \sigma_q(n-2t,k-t) &\text{if }\type(q,n)=\text{(\NNA)
             and } \textbf{1}\not\in U,\\[.5ex]
       \sigma_q(n-2t+1,k-t) &\text{if }\type(q,n)= \text{(\NNA) and } \textbf{1}\in U.\\
       \end{array}\right.
\]
\end{proposition}

One may note that if $\type(q,n)\in\{\text{(P), (\NOne)}\}$ then, by definition, $\type(q,n-2t)=\type(q,n)$.
The second part of Proposition~\ref{prop:typeqn}
implies $\smash{w(\F_q^{n-2t}, \cdot) = w(\F_q^n, \cdot)-t}$, which together with Theorem~\ref{thm:Sigma}(a)
gives $\tau_q(n-2t,k-t,w-t)=\sigma_q(n-2t,k-t)$.
This is not necessarily true if $\type(q,n)\in\{\text{(H), (E)}\}$, which is why we introduced the function $\tau_q$.

\begin{proof}
Throughout the proof we use $B$ for the dot product.
By \cref{cor:SOCodesContaining}, the desired cardinality is given by
$\tilde{\sigma}_{k,U}=\sigma(U^\perp/U,B_U,k-t)$, and, thanks to \cref{thm:Sigma}, the value of $\sigma(U^\perp/U,B_U,k-t)$ only depends on the type of the bilinear space $(U^\perp/U,B_U)$.
Note that $\dim U^\perp/U=n-2t$. Moreover, the Witt index of $\smash{(U^\perp/U,B_U)}$ is given by $w-t$; see \cref{prop:quotient}.
\\
(a) Let $\type(q,n)\in\{\text{(P), (H), (E), (\NOne)}\}$. Then \cref{prop:quotient} tells us that $(U^\perp/U,B_U)$ is of the same type as $(\F_q^n,B)$.
Thanks to \cref{thm:Sigma}(a) we thus have
\[
 \tilde{\sigma}_{k,U}
 =\prod_{i=1}^{k-t}\frac{(q^{n-t-w-i}+1)(q^{w-t-i+1}-1)}{q^i-1}
 =\tau_q(n-2t,k-t,w-t).
\]
(b) Let $\type(q,n)=$ (\NNA) and  $\textbf{1}\not\in U$.
Then \cref{ex:InducedAlternating} tells us that the induced bilinear space
$(U^\perp/U,B_U)$ is also of type (\NNA).
Hence $(U^\perp/U,B_U)$ has the same dimension and type as
$(\F_q^{n-2t},B)$, and therefore
$\smash{\tilde{\sigma}_{k,U}=\sigma_q(n-2t,k-t)}$.
\\
(c) Let $\type(q,n)=$ (\NNA) and $\textbf{1}\in U$.
By \cref{ex:InducedAlternating} the induced form $B_U$ is alternating, and \cref{thm:Sigma}(b) leads to
$\smash{\tilde{\sigma}_{k,U}=\prod_{i=1}^{k-t}(q^{n-2t-2i+2}-1)/(q^i-1)}$.
But the latter equals $\sigma_q(n-2t+1,k-t)$, as can be seen from \cref{P-sigmaqnk} since $\type(q,n-2t+1)=$ (\NOne).
\end{proof}

If $t=1$, the last case of the previous proposition amounts to $U=\subspace{\textbf{1}}$, and thus $\tilde{\sigma}_{k,U}$ is the number of $k$-dimensional self-orthogonal spaces~$U$ such that $B_U$ is alternating. The result agrees with  \cref{prop:NonAltInduced}.

Proposition~\ref{prop:SOContainx} allows us to explicitly compute the aggregate weight distribution of the set of self-orthogonal
$[n,k]_q$-codes.
In the following we set $\tau_q(n-2,\ell,w)=1=\sigma_q(n-2,\ell)$ if $n=2$.

\begin{corollary}\label{cor:sumweight}
Let $j>0$ and $1\leq k\leq w:=w(\F_q^n,\,\cdot\,)$.
\begin{alphalist}
\item Let $\type(q,n)\in\{\text{(P),\,(\NOne),\,(H),\,(E)}\}$.
      Then
         \[
              A_q^{(j)}(n,k,k)=\zeta_q(n,j)\tau_q(n-2,k-1,w-1).
         \]
\item Let $\type(q,n)=$ (\NNA) and $j<n$.
        Then
         \[
              A_q^{(j)}(n,k,k)=\zeta_q(n,j)\sigma_q(n-2,k-1).
         \]
\item Let $\type(q,n)=$ (\NNA) and $j=n$.
      Then
        \[
              A_q^{(j)}(n,k,k)
              =\zeta_q(n,n)\sigma_q(n-2,k-1)+(q-1)\big(\sigma_q(n-1,k-1)-\sigma_q(n-2,k-1)\big).
         \]
\end{alphalist}
Note that in all cases we have $\smash{A_q^{(j)}(n,1,1)=\zeta_q(n,j)}$, as it should be.
\end{corollary}

\begin{proof}
(a) and (b) follow immediately from \eqref{note} and \cref{prop:SOContainx} for $t=1$.
For (c) note that there are exactly $q-1$ nonzero vectors~$v$ in $\subspace{\textbf{1}}$ and $\zeta_q(n,n)-q+1$ self-orthogonal vectors~$v$ of weight~$n$ not in this subspace.
Again, the result follows from \cref{prop:SOContainx}.
\end{proof}

We can finally give a closed formula for the aggregate weight distribution of the set of $\ell$-complementary $[n,k]_q$-codes.
The proof again uses M\"obius inversion to reduce the problem to the self-orthogonal case, which we solved in Corollary~\ref{cor:sumweight}.
In the following we set $\smash{\Gaussian{a}{b}_q=0}$ if $a<0$ or $b<0$.

\begin{theorem}\label{T-AvgWeight}
Suppose $n \ge 3$ and $1\leq k\leq n$. Let $w=w(\F_q^n,\,\cdot\,)$. Denote the
\textbf{Krawtchouk coefficients} by
\[
    K_q(n,i,j):=\sum_{r=0}^i (-1)^r (q-1)^{i-r} \binom{j}{r} \binom{n-j}{i-r}
    \ \text{ for }0\leq i,j\leq n;
\]
see e.g.~\cite[p.~76]{huffman2010fundamentals}. Then for all $i\geq 0$
\[
   A_q^{(i)}(n,k,\ell) = \sum_{s=\ell}^w \GaussianD{s}{\ell}  (-1)^{s-\ell} \, q^{\binom{s-\ell}{2}} B(s),
\]
where, for $\ell \le s \le w$,
\[
   B(s)=
   \left(\GaussianD{n\!-\!2s}{k\!-\!s}_q-\GaussianD{n\!-\!2s\!-\!1}{k\!-\!s\!-\!1}_q\right) A_q^{(i)}(n,s,s)  + q^{-s} \GaussianD{n\!-\!2s\!-\!1}{k\!-\!s\!-\!1}_q  \, \sum_{j=0}^n A_q^{(j)}(n,s,s)  K_q(n,i,j).
\]
\end{theorem}


Note that $\smash{A_q^{(i)}(n,k,\ell)=0}$ for $\ell>w(\F_q^n,\,\cdot\,)$, as it has to be.
Furthermore, for $i=0$ the above leads to the special case
\[
   A_q^{(0)}(n,k,\ell)=\sum_{s=\ell}^w\GaussianD{s}{\ell}  (-1)^{s-\ell} \, q^{\binom{s-\ell}{2}}\GaussianD{n-2s}{k-s}\sigma_q(n,s),
\]
and by \cref{thm:sigmakell} the right hand side equals $\sigma_q(n,k,\ell)$, all of which confirms $A_q^{(0)}(n,k,\ell)=\sigma_q(n,k,\ell)$.
In this sense the above may be regarded as a generalization of \cref{thm:sigmakell} to the full aggregate weight distribution.

\begin{proof}
Fix $v \in \F_q^n$ of weight $i\geq 0$.
Let $\mL$ denote the lattice of subspaces of $\F_q^n$. Define functions $f,g:\mL \to \Z$ by
\begin{align*}
    f(V) &=\big|\{C \le \F_q^n \mid \dim(C)=k, \, C \cap C^\perp=V, \, v \in C \}\big|, \\[.7ex]
    g(V) &= \sum_{\substack{W \in \mL \\ W \ge V}} f(W).
\end{align*}
By definition,
$g(V)$ counts the number of $k$-dimensional codes $C \le \F_q^n$ with
$V \le C \le V^\perp$ and $v \in C$. Therefore,
for all $V \in \mL$ of dimension $s$ we have
\[
g(V) = \left\{\begin{array}{cl}
\hfil 0 & \mbox{if $V$ is not self-orthogonal or $v \notin V^\perp$,} \\[.7ex]
\GaussianD{n-s-\dim(V+\langle v\rangle)}{k-\dim(V+\langle v\rangle)}_q & \mbox{if $V$ is self-orthogonal and $v \in V^\perp$.}
\end{array}\right.
\]
By applying M\"obius inversion in the lattice $\mL$,
we obtain that the number of
$\ell$-complementary codes $C \le \F_q^n$ of dimension $k$ that contain $v$ is
\begin{equation} \label{interm}
 \sum_{s=0}^n \sum_{\substack{W\in\Sigma_q(n,s) \\ v \in W^\perp}} \GaussianD{s}{\ell} \GaussianD{n-s-\dim(W+ \langle v \rangle)}{k-\dim(W+ \langle v \rangle)} \, (-1)^{s-\ell} \, q^{\binom{s-\ell}{2}}.
\end{equation}
Therefore, by definition,
\begin{align}
    A_q^{(i)}(n,k,\ell)
&= \sum_{\substack{v \in \F_q^n \\ \wt(v)=i}}\sum_{s=0}^n \, \sum_{\substack{W\in\Sigma_q(n,s) \\ v \in W^\perp}} \GaussianD{s}{\ell} \GaussianD{n-s-\dim(W+ \langle v \rangle)}{k-\dim(W+ \langle v \rangle)} \, (-1)^{s-\ell} \, q^{\binom{s-\ell}{2}} \nonumber \\
&= \sum_{s=0}^n \GaussianD{s}{\ell}  (-1)^{s-\ell} \, q^{\binom{s-\ell}{2}}
   \sum_{W\in\Sigma_q(n,s)} \  \sum_{\substack{v \in \F_q^n \\ \wt(v)=i \\ v \in W^\perp}}  \GaussianD{n-s-\dim(W+ \langle v \rangle)}{k-\dim(W+ \langle v \rangle)}. \label{ff1}
\end{align}
Now observe that
for a given self-orthogonal $W \in \mL$ of dimension $s$ we have
\begin{align}
    \sum_{\substack{v \in \F_q^n \\ \wt(v)=i \\ v \in W^\perp}}  \GaussianD{n-s-\dim(W+ \langle v \rangle)}{k-\dim(W+ \langle v \rangle)} \nonumber
&=\sum_{\substack{v \in \F_q^n \\ \wt(v)=i \\ v \in W \cap W^\perp}}  \GaussianD{n-2s}{k-s} \; +
\sum_{\substack{v \in \F_q^n \\ \wt(v)=i \\ v \in W^\perp \setminus W}}  \GaussianD{n-2s-1}{k-s-1} \nonumber \\
&\hspace*{-5em}= \left(\GaussianD{n-2s}{k-s}_q - \GaussianD{n-2s-1}{k-s-1}_q \right)  A^{(i)}(W) +
\GaussianD{n-2s-1}{k-s-1}_q  A^{(i)}(W^\perp), \label{ff2}
\end{align}
where we use that $W\cap W^\perp=W$. Using the MacWilliams identities (see e.g.~\cite[p.~257]{huffman2010fundamentals}) we write
\begin{equation} \label{ff3}
A^{(i)}(W^\perp) = q^{-s} \sum_{j=0}^n A^{(j)}(W)K_q(n,i,j).
\end{equation}
Now combining \eqref{ff1},~\eqref{ff2} and \eqref{ff3} leads to the desired result.
\end{proof}

We conclude by observing that the expression for $B(s)$ in the statement of Theorem~\ref{T-AvgWeight} simplifies according to
\[
   \GaussianD{n-2s}{k-s}_q-\GaussianD{n-2s-1}{k-s-1}_q
   =\left\{\begin{array}{cl}1&\text{if }s=k=n/2,\\[1ex]
          q^{k-s}\GaussianD{n-2s-1}{k-s}_q& \text{otherwise.}
     \end{array}\right.
\]

\section{Asymptotic Behavior and Qualitative Comparisons}
\label{sec:5}

In this section we will investigate the asymptotic behavior of various invariants of self-orthogonal codes as $q\rightarrow\infty$.
We will then compare these invariants, in the asymptotic regime, to those of unrestricted (i.e., not necessarily self-orthogonal) codes. Our results show that self-orthogonal codes, despite being very sparse within the set of all $[n,k]_q$-codes, behave quite similarly to unrestricted codes when $q$ is large.
More evidence of this will be given in Section~\ref{sec:6}.
For any functions $f,g:\N\rightarrow\R$ we define
\[
    f\sim g:\Longleftrightarrow \lim_{q\rightarrow\infty}\frac{f(q)}{g(q)}=1.
\]
All asymptotic estimates and limits in this section are for $q \rightarrow \infty$, unless otherwise specified.
Occasionally we will obtain different asymptotic behavior depending on the type of the bilinear space.
In that case $\sim$ is understood for $q\rightarrow\infty$ subject to $q$ even, $q\equiv1\mmod{4}$, or $q\equiv3\mmod{4}$, depending on the type.
We will repeatedly use the well-known asymptotic estimate
\[
   \GaussianD{a}{b}_q \sim q^{b(a-b)}
\]
for integers $a \ge b \ge 0$.

We start with the following result about the asymptotic proportion of self-orthogonal subspaces in a bilinear space $(V,B)$.

\begin{theorem}\label{thm:SigmaAsymp}
Let $n\geq2$ and $(V,B)$ be a bilinear space as in \cref{N-VB} and let $1\leq k\leq w(V,B)$.
The asymptotic proportion of the $k$-dimensional self-orthogonal subspaces in~$V$ within the set of all $k$-dimensional subspaces
is given by
\[
\frac{\sigma(V,B,k)}{\Gaussian{n}{k}_q}\sim\left\{\begin{array}{cl}
       q^{-\binom{k+1}{2}}&\text{ for type (P), (\NOne), (\NNA), (E) and for (H) if $k<n/2$},\\
       2q^{-\binom{k+1}{2}}&\text{ for type (H) if }k=n/2,\\
       q^{k-\binom{k+1}{2}}&\text{ for type (\NA)}.
      \end{array}\right.
\]
\end{theorem}

\begin{proof}
We consider the expressions for $\sigma(V,B,k)$ given in \cref{thm:Sigma}.
The asymptotic estimate $\smash{a_k\sim q^{\binom{k+1}{2}}}$ is clear. We analyze the various cases separately.
\\
1) Type (P) and (\NOne):  Using the fact that  $\prod_{i=1}^{k}(q^{n-2i+1}-1)\sim q^{k(n+1)-k(k+1)}=q^{k(n-k)}$, we obtain
\[
     \frac{\sigma(V,B,k)}{\Gaussian{n}{k}_q} \sim q^{k(n-k)-\binom{k+1}{2}-k(n-k)}=q^{-\binom{k+1}{2}}.
\]
2) Type (H) and (E): First of all, $\prod_{i=1}^{k-1}(q^{n-2i}-1)\sim q^{(k-1)n-(k-1)k}=q^{(k-1)(n-k)}$.
Set $\varepsilon=1$ for type (H) and $\varepsilon=-1$ for type (E). Then
\[
    q^{n-k}+\varepsilon q^{n/2}-\varepsilon q^{n/2-k}-1
       \sim\left\{\begin{array}{cl}q^{n-k}&\text{if }k<n/2,\\  2q^k&\text{if $k=n/2$ and $\varepsilon=1$}.
       \end{array}\right.
\]
Recall that $\varepsilon=-1$ together with $k=n/2$ does not occur.
Now all statements for these two types follow from \cref{thm:Sigma}.
\\
3) Type (\NNA): In this case $\smash{\sigma(V,B,k)\sim q^{n-k+(k-1)(n-k)-\binom{k+1}{2}}=q^{k(n-k)-\binom{k+1}{2}}}$ and the desired result follows.
\\
4) Type (\NA): The asymptotic estimate $\smash{\sigma(V,B,k)\sim q^{k(n-k+1)-\binom{k+1}{2}}}$ leads to the desired result.
\end{proof}

From now on we focus solely on the bilinear space $(\F_q^n,\,\cdot\,)$.
For that space the above results read as follows.
The version below, providing the asymptotic behavior of the cardinality of $k$-dimensional self-orthogonal codes in $\F_q^n$, rather than their proportion,
will be useful later on.
Recall the terminology from \cref{prop:typeqn}.

\begin{corollary}\label{cor:SigmaDotAsymp}
For the bilinear space $(\F_q^n,\,\cdot\,)$ and $1\leq k\leq w(\F_q^n,\,\cdot\,)$ we have
\begin{align*}
   \sigma_q(n,k)\sim \left\{\begin{array}{cl}
       2q^{k(n-k)-\binom{k+1}{2}} &\text{if $\type(q,n)=$ (H) and $k=n/2$},\\[1ex]
       q^{k(n-k)-\binom{k+1}{2}} & \text{otherwise.}
    \end{array}\right.
\end{align*}
\end{corollary}

Our next goal is to determine the asymptotic behavior of the average weight distribution of self-orthogonal codes as $q\rightarrow\infty$. We will need the asymptotics of $\tau_q(n,k,w)$, defined in \cref{nota:tau}. Unsurprisingly, it agrees with the one for $\sigma_q(n,k)$.

\begin{proposition}\label{prop:AsympTau}
Let $w\in\{1,\ldots,n/2\}$ and $k\in\{1,\ldots,w\}$. Then
\[
   \tau_q(n,k,w)\sim \left\{\begin{array}{cl}
       2q^{k(n-k)-\binom{k+1}{2}} &\text{if $n$ even and $k=w=n/2$},\\[1ex]
        q^{k(n-k)-\binom{k+1}{2}} & \text{otherwise.}
    \end{array}\right.
\]
\end{proposition}

\begin{proof}
Consider~$\tau_q(n,k,w)$ as defined in \eqref{e-tau}. The numerator of the $i$-th factor is $N(i):=q^{n-2i+1}+q^{w-i+1}-q^{n-w-i}-1$. It is dominated by $q^{n-2i+1}$ in all cases except for $i=k=w=n/2$.
In that case $N(i)=2q-2\sim 2q=2q^{n-2i+1}$, while in all other cases $N(i)\sim q^{n-2i+1}$.
Now $\tau_q(n,k,w)=\prod_{i=1}^k N(i)/(q^i-1)$ leads to the desired result.
\end{proof}

We also need the asymptotic behavior of $\zeta_q(n,i)$, defined in \cref{nota:zeta}.

\begin{proposition}\label{prop:AsympZeta}
We have $\zeta_q(n,1)=0$ and
\[
   \zeta_q(n,2)=
    \left\{\begin{array}{cl}\dbinom{n}{2}(q-1) &\text{if $q$ even},\\[1.9ex]
           2\dbinom{n}{2}(q-1) &\text{if $q\equiv1\mmod{4}$},\\[.7ex]
              0 &\text{if $q\equiv3\mmod{4}$.}\end{array}\right.
\]
For $i>2$ we have $\zeta_q(n,i)\sim \dbinom{n}{i}q^{i-1}$.
\end{proposition}

\begin{proof}
Clearly, there are no self-orthogonal vectors of Hamming weight~$1$. For $i>1$ recall the expression for $\zeta_q(n,i)$ in \cref{prop:zeta}.
In order to evaluate it, we need $\sigma_q(j,1)$. From \cref{P-sigmaqnk} we know
\[
   \sigma_q(j,1)=\left\{\begin{array}{cl} \displaystyle  \frac{q^{j-1}-1}{q-1}&\text{if $\type(q,j)\in\{\text{(P), (\NOne), (\NNA)}\}$},\\[12pt]  
       \displaystyle \frac{q^{j-1}+\varepsilon q^{j/2}-\varepsilon q^{j/2-1}-1}{q-1} &\text{if $\type(q,j)\in\{\text{(H),(E)}\}$},
     \end{array}\right.
\]
where $\varepsilon=1$ if $\type(q,j)=$ (H) and $\varepsilon=-1$ if $\type(q,j)=$ (E).
As a consequence,
\begin{equation}\label{e-middleterm}
  (q-1)\sigma_q(j,1)+1=\left\{\begin{array}{cl}q^{j-1} &\text{if $\type(q,j)\in\{\text{(P), (\NOne), (\NNA)}\}$},\\ 
       q^{j-1}+\varepsilon q^{j/2}-\varepsilon q^{j/2-1} &\text{if $\type(q,j)\in\{\text{(H),(E)}\}$}.
  \end{array}\right.
\end{equation}
For $j>2$ we have $j/2<j-1$ and thus $(q-1)\sigma_q(j,1)+1\sim q^{j-1}$.
For $j=2$ we have
\[
  (q-1)\sigma_q(2,1)+1=\left\{\begin{array}{cl} q &\text{if $q$ even},\\ 2q-1 &\text{if $q\equiv1\mmod{4}$},\\1 &\text{if $q\equiv3\mmod{4}$.}\end{array}\right.
\]
Evaluating the expression for $\zeta_q(n,2)$ in \cref{prop:zeta} we obtain the desired result for $i=2$.
For $i>2$ we proceed as follows. Equation~\eqref{e-middleterm} reads
$(q-1)\sigma_q(j,1)+1=q^{j-1}+\alpha_q(j)q^{j/2-1}(q-1)$, where $\alpha_q(j)=0$ if $\type(q,j)\in\{\text{(P), (\NOne), (\NNA)}\}$ and
$\alpha_q(j)=\varepsilon$ if $\type(q,j)\in\{\text{(H), (E)}\}$.
Now we compute

\begin{align*}
   \frac{\zeta_q(n,i)}{\binom{n}{i}}&=(-1)^i+\sum_{j=1}^i\binom{i}{j}(-1)^{i-j}\Big(q^{j-1}+\alpha_q(j)q^{j/2-1}(q-1)\Big)\\
       &=(-1)^i+\sum_{j=1}^i\binom{i}{j}(-1)^{i-j}q^{j-1}+\frac{q-1}{q}\sum_{j=1}^i\binom{i}{j}(-1)^{i-j}\alpha_q(j)q^{j/2}\\
       &=(-1)^i+\frac{1}{q}\Big((q-1)^i-(-1)^i\Big)+\frac{q-1}{q}\sum_{j=1}^i\binom{i}{j}(-1)^{i-j}\alpha_q(j)q^{j/2}\\
       &=\frac{(q-1)^i-(-1)^i+(-1)^iq}{q}+\frac{q-1}{q}\sum_{j=1}^i\binom{i}{j}(-1)^{i-j}\alpha_q(j)q^{j/2}.
\end{align*}
Since the first term has degree $i-1$ and the second one has degree at most $i/2<i-1$, we conclude $\zeta_q(n,i)\sim\binom{n}{i}q^{i-1}$, as stated.
\end{proof}

Now we are ready to consider the average weight distribution of self-orthogonal codes.
Note that $\smash{\overline{A}_q^{(1)}(n,k,k)=0}$.

\begin{theorem}\label{thm:AsympAvgWeight}
Let $2\leq j\leq n$ and $1\leq k\leq w:=w(\F_q^n,\,\cdot\,)$.
Then
\begin{equation}\label{e-GenCase}
    \overline{A}_q^{(j)}(n,k,k)\sim\binom{n}{j}q^{j-n+k}
\end{equation}
with the following exceptions for $j=2$:
\begin{equation}\label{e-qodd}
   \overline{A}_q^{(2)}(n,k,k)=0\ \text{ if }q\equiv3\mmod{4}
   \ \ \text{ and } \
   \overline{A}_q^{(2)}(n,k,k)\sim2\binom{n}{2}q^{2-n+k}\ \text{ if }\ q\equiv1\mmod{4}.
\end{equation}
\end{theorem}

\begin{proof}
Recall from \cref{def:AvgDistr} that $\smash{\overline{A}_q^{(j)}(n,k,k)=A_q^{(j)}(n,k,k)/\sigma_q(n,k)}$, and that $\smash{A^{(j)}_q(n,k,k)}$ is provided in \cref{cor:sumweight}.
The case $j=n=2$, thus $k=1$, follows immediately from \cref{cor:sumweight} and \cref{prop:AsympZeta}.
Thus let $n\geq3$.
We proceed with the various types.
\\
1) \underline{\smash{$\type(q,n)=$ (P), (\NOne), (H), (E):}} Then
$A^{(j)}_q(n,k,k)=\zeta_q(n,j)\tau_q(n-2,k-1,w-1)$.
For types (P) and (\NOne) we have $w=(n-1)/2$, while for type (E) we have $w=n/2-1$.
In either case $w-1\neq(n-2)/2$ and thus the asymptotics of $\tau_q(n-2,k-1,w-1)$
is given by the second case in \cref{prop:AsympTau}.
For type (H) we have $w=n/2$ and the asymptotics is determined the first case if $k=w$ and by the second one otherwise.
For $j>2$ all of this together with \cref{prop:AsympZeta} results in
\begin{equation}\label{e-AjqSim}
    A_q^{(j)}(n,k,k)\sim\alpha\binom{n}{j}q^{j-1+(k-1)(n-k-1)-\binom{k}{2}},
\end{equation}
where $\alpha=2$ if $k=n/2$ and $\type(q,n)=$ (H), and $\alpha=1$ otherwise. Dividing by $\sigma_q(n,k)$ and using \cref{cor:SigmaDotAsymp} we arrive at the stated formula.
For $j=2$ \cref{prop:AsympZeta} immediately leads to \eqref{e-qodd}. 
\\
2) \underline{\smash{$\type(q,n)=$ (\NNA):}}
In this case $\smash{\zeta_q(n,j)\sim\binom{n}{j}q^{j-1}}$ for all $j\geq 2$.
Furthermore, if $j<n$, then $A^{(j)}_q(n,k,k)=\zeta_q(n,j)\sigma_q(n-2,k-1)$ thanks to \cref{cor:sumweight}(b).
By \cref{cor:SigmaDotAsymp} and \cref{prop:AsympTau} the asymptotics of $\sigma_q(n-2,k-1)$ agrees with that of $\tau_q(n-2,k-1,w-1)$ and as in~1) we obtain \eqref{e-GenCase}.
If $j=n$ then $\smash{A^{(j)}_q(n,k,k)}$ is given in \cref{cor:sumweight}(c).
Using the asymptotic estimates $\smash{\sigma_q(n-2,k-1)\sim q^{(k-1)(n-k-1)-\binom{k}{2}}}$ and $\smash{\sigma_q(n-1,k-1)\sim q^{(k-1)(n-k)-\binom{k}{2}}}$ we observe that
$\smash{A^{(j)}_q(n,k,k)}\sim\zeta_q(n,n)\sigma_q(n-2,k-1)\sim q^{k(n-k)-\binom{k}{2}}$, which again leads to
\eqref{e-GenCase}.
\end{proof}

It is natural to compare 
the average weight distribution of a self-orthogonal code with that of an ``unrestricted'' code. We include this comparison in the following remark.

\begin{remark}\label{rem:AsympAvgWeightGeneral}
Define

\[
    B_q^{(j)}(n,k):=\sum_{\genfrac{}{}{0pt}{}{x\in\F_q^n}{\wt(x)=j}}\big|\{C\leq\F_q^n\mid C \text{ is an $[n,k]_q$-code and }x\in C\}\big|
\]
and $\smash{\overline{B}_q^{(j)}(n,k):=B_q^{(j)}(n,k)/\Gaussian{n}{k}_q}$.
Then the tuple $\smash{(\overline{B}_q^{(j)}(n,k)\mid 0\leq j\leq n)}$ is the average weight distribution of an $[n,k]_q$-code.
Using that there are $\binom{n}{j}(q-1)^j$ vectors of weight~$j$ and $\Gaussian{n-1}{k-1}_q$ codes of dimension~$k$ containing a
given nonzero vector, one easily obtains the asymptotic estimate $\smash{B_q^{(j)}(n,k)=\binom{n}{j}(q-1)^j\Gaussian{n-1}{k-1}_q\sim\binom{n}{j}q^{j+(k-1)(n-k)}}$.
As a consequence,
\[
   \overline{B}_q^{(j)}(n,k)\sim\binom{n}{j}q^{j-n+k}\ \text{ for all $n$ and $k$}.
\]
This shows that for even $q$, the average weight distribution of a self-orthogonal $[n,k]_q$-code behaves asymptotically just like the
average weight distribution of a general $[n,k]_q$-code.
Only for odd~$q$, the asymptotic behavior of $\smash{\overline{A}_q^{(2)}(n,k,k)}$ in \eqref{e-qodd} allows us to
distinguish between self-orthogonal codes and general codes.
\end{remark}

\section{Existence Results for Self-Orthogonal Codes}
\label{sec:6}

In this section we turn to the most studied parameter of an error-correcting code, namely its minimum distance; see \cref{nota:codes}.
An error-correcting code cannot have dimension and minimum distance both large at the same time. Indeed, the Singleton Bound states that $\smash{k \le n-d(C)+1}$ for any nonzero $[n,k]_q$-code~$C$; see~\cite{singleton1964maximum}. Codes attaining the Singleton Bound with equality are called \textbf{maximum distance separable} (\textbf{MDS}) and form a central topic in coding theory; see for instance~\cite{macwilliams1977theory}.

We wish to estimate the number of $k$-dimensional self-orthogonal codes $C \le \F_q^n$ with minimum Hamming distance at least $d$ over a large finite field. The main two results of this section (Theorems~\ref{thm:asympdense} and~\ref{thm:asympdense2} below) compute the asymptotic density
of the $k$-dimensional, non-MDS self-orthogonal codes.
In particular, they prove that
a uniformly random $k$-dimensional self-orthogonal code is MDS with probability approaching 1 as $q$ grows.

We need the following notation and results. Recall \cref{def:Hamming}.

\begin{notation}\label{nota:FqS}
For a subset $S \subseteq \{1,\ldots,n\}$ we define $\F_q^n(S)=\{v\in\F_q^n\mid \sigma(v)\subseteq S\}$, that is, $\F_q^n(S)$ is the space of vectors in $\F_q^n$
whose support is contained in $S$.
\end{notation}

\begin{lemma} \label{lem:prr2}
Let $w=w(\F_q^n,\,\cdot\,)$.
For $k\in\{1,\ldots,w\}$ and $i\in\{0,\ldots,k\}$ set
\[
   \delta_q(n,k,i)=\left\{\begin{array}{cl}
   \sigma_q(n-2i,k-i) &\text{if }\type(q,n)=\text{(\NNA)},\\
   \tau_q(n-2i,k-i,w-i) &\text{if }\type(q,n)\in\{\text{(P), (H), (E), (\NOne)}\},
   \end{array}\right.
\]
where $\tau_q(n-2i,k-i,w-i)$ is as in \cref{nota:tau}.
Let $S \subseteq \{1,\ldots,n\}$ be a set of size $1 \le t< n$.
Set $\hat{w}=w(\F_q^t,\,\cdot\,)$.
Then
\[
    \big|\big\{C \in \Sigma_q(n,k) \mid  C \cap \F_q^n(S) \neq \{0\}\big\}\big|=\sum_{i=1}^{\min\{k,\hat{w}\}} \sigma_q(t,i) \, \delta_q(n,k,i) \, (-1)^{i-1} q^{\binom{i}{2}}.
\]
\end{lemma}

\begin{proof}
Let $D=\F_q^n(S)$ and $\mL$ be the lattice of subspaces of $D$.
Denote its M{\"o}bius function by~$\mu_{\mL}$.
Consider the quantity
\[
\Delta=\sum_{\substack{C \in \Sigma_q(n,k)}}  \;
\sum_{\substack{U \in \mL \\ U \le C}} \mu_\mL(\{0\},U).
\]
Note that for a subspace $U \le \F_q^n$ we have $U \in \mL$ and $U \le C$ if and only if $U \le C \cap D$.
Therefore the properties of the M\"obius function imply
\begin{equation} \label{delta1}
\Delta=\sum_{\substack{C \in \Sigma_q(n,k)}}  \;
\sum_{U \le D \cap C} \mu_\mL(\{0\},U) = |\{C \in \Sigma_q(n,k) \mid  C \cap D =\{0\}\}|.
\end{equation}
On the other hand, exchanging the summation order in the definition of $\Delta$ one obtains
\begin{equation} \label{delta2}
     \Delta=\sum_{U \in \mL} \mu_\mL(\{0\},U) \cdot |\{C \in \Sigma_q(n,k) \mid U \le C \}|.
\end{equation}
Let $U\in\mL$ be such that $\dim U=i\in\{1,\ldots,k\}$. Since $t<n$, the all-one vector~$\textbf{1}$ is not in~$U$ and thus  \cref{prop:SOContainx} implies
\begin{equation} \label{eq:meaning_delta}
    |\{C \in \Sigma_q(n,k) \mid U \le C \}| = \left\{\begin{array}{cl}
0 & \mbox{if $U$ is not self-orthogonal,} \\
\delta_q(n,k,i)
& \mbox{if $U$ is self-orthogonal,}
\end{array}\right.
\end{equation}
with $\delta_q(n,k,i)$ as in the statement.
Clearly, the second case only arises if $i\leq \hat{w}$.
In that case the number of $i$-dimensional self-orthogonal spaces $U \in \mL$ is $\sigma_q(t,i)$ and thus \eqref{delta2} leads to
\[
     \Delta = \sum_{i=0}^{\min\{k,\hat{w}\}} \sigma_q(t,i) \, \delta_q(n,k,i)\, (-1)^iq^{\binom{i}{2}}.
\]
Now $\big|\big\{C \in \Sigma_q(n,k) \mid  C \cap \F_q^n(S) \neq \{0\}\big\}\big|=\sigma_q(n,k)-\Delta$, and using that $\delta_q(n,k,0)=\sigma_q(n,k)$ (see Notation~\ref{nota:tau}), we arrive at the stated result.
\end{proof}

The following lemma will be needed to deal with the last case in \cref{prop:SOContainx}.

\begin{lemma}\label{lem:S1S2}
Let $q$ be even.
Let $S_1,\,S_2 \subseteq \{1,\ldots,n\}$ be distinct sets of size $1 \le t< n$.
Let~$A$ be a set of distinct representatives of the nonzero self-orthogonal vectors in $\F_q^n$, up to scalar multiples.
Let $A_i\subseteq A$ be the subsets representing the nonzero self-orthogonal elements in~$\F_q^n(S_i)$ for $i=1,2$.
Then
\[
   |\{(v_1,v_2)\mid v_i\in A_i,\, \textbf{1}\in\subspace{v_1,v_2}\}|\leq q^{\ell},
\]
where $\ell=2t-n-1$ if $|S_1\cap S_2|>0$ and $\ell=0$ if $|S_1\cap S_2|=0$.
\end{lemma}

\begin{proof}
Clearly $\textbf{1}\in\subspace{v_1,v_2}$ implies $S_1\cup S_2=[n]$, and thus $\hat{\ell}:=|S_1\cap S_2|=2t-n\geq0$.
Furthermore,  since $S_1$ and~$S_2$ are distinct, we have $\textbf{1}\in\langle v_1,v_2\rangle\Longleftrightarrow v_2\in \subspace{\textbf{1},v_1}$.
The existence of some $v_2\in A_2\cap\subspace{\textbf{1},v_1}$ implies that~$v_1$ is constant and nonzero  on $S_1\setminus S_2$.
Furthermore, if~$v_1$ is of that form, then $v_2\in A_2\cap\subspace{\textbf{1},v_1}$ is unique.
It remains to find an upper bound for the number of self-orthogonal vectors~$v_1\in\F_q^n(S_1)$ that are nonzero constant on $S_1\setminus S_2$.
\\
i) If~$\hat{\ell}=0$, then $S_1\setminus S_2=S_1$, and there are $q-1$ such vectors.
\\
ii) Let $\hat{\ell}>0$. Assume without loss of generality that $1\in S_1\cap S_2$.
Then we have  $q-1$ options for the value on $S_1\setminus S_2$ and $\smash{q^{\hat{\ell}-1}}$ options on $(S_1\cap S_2)\setminus\{1\}$ and at most~$1$ choice for the remaining entry to achieve self-orthogonality (recall  that~$q$ is even).
Thus there are at most $\smash{(q-1)q^{\hat{\ell}-1}}$ self-orthogonal vectors~$v_1\in\F_q^n(S_1)$ with the desired property.
\\
Accounting for scalar multiples of~$v_1$, we arrive at the stated upper bound.
\end{proof}

We are now ready to compute the asymptotic behavior of the density function of self-orthogonal codes with small minimum distance.

\begin{theorem}\label{thm:asympdense}
Let $4 \le d < n$, and $1 \le k \le \min\{n-d,w(\F_q^n,\,\cdot\,)\}$.
Then
\[
\frac{|\{C \in \Sigma_q(n,k) \mid d(C) \le d-1\}|}{\sigma_q(n,k)} \sim \binom{n}{d-1} q^{d+k-n-2} \quad \mbox{ as $q \rightarrow \infty$}.
\]
\end{theorem}

\begin{proof}
1) Let $S_1,...,S_L$ be the $(d-1)$-subsets of $\{1,\ldots,n\}$. For a non-empty $J \subseteq \{1,\ldots,L\}$ we define
\[
\Gamma_J=\{C \in \Sigma_q(n,k) \mid C \cap \F_q^n(S_j) \neq \{0\} \mbox{ for all $j \in J$}\}
\]
and set $\smash{\Gamma_j:=\Gamma_{\{j\}}}$ for $j \in \{1,\ldots,L\}$.
Note that for any $J\subseteq\{1,\ldots,L\}$ we have $\smash{\bigcap_{j \in J}\Gamma_j=\Gamma_J}$.
By definition,
\[
\{C \in \Sigma_q(n,k) \mid d(C) \le d-1\} = \bigcup_{j=1}^L \Gamma_j.
\]
By the inclusion-exclusion principle,
\begin{equation}\label{e-inclexcl}
|\{C \in \Sigma_q(n,k) \mid d(C) \le d-1\}| =
\sum_{j=1}^L |\Gamma_j| - \sum_{|J|=2} |\Gamma_J|+\sum_{|J|=3}|\Gamma_J|-\ldots.
\end{equation}
In the rest of the proof we will examine the asymptotics of the right hand side of~\eqref{e-inclexcl}. It will turn out that it is given by the asymptotics of the first sum.
For that, the assumption $k\leq n-d$ will be needed.
\\[.5ex]
2) We evaluate the first sum on the RHS of~\eqref{e-inclexcl}.
By Lemma~\ref{lem:prr2}, it is given by
\begin{equation} \label{sum1}
\sum_{j=1}^L |\Gamma_j| = \binom{n}{d-1} \, \sum_{i=1}^{\min\{\hat{w},k\}} \sigma_q(d-1,i) \, \delta_q(n,k,i) \, (-1)^{i-1} q^{\binom{i}{2}},
\end{equation}
where $\delta_q(n,k,i)$ is as in that lemma and $\hat{w}=w(\F_q^{d-1},\,\cdot\,)$.
We now turn to its asymptotics.
Note that by \cref{prop:AsympTau} and \cref{cor:SigmaDotAsymp} we have
\begin{equation}\label{e-deltaAsymp1}
     \delta_q(n,k,i)\sim c q^{(k-i)(n-k-i)-\binom{k-i+1}{2}},
\end{equation}
where
\begin{equation}\label{e-deltaAsymp2}
     c=\left\{\begin{array}{cl}2&\text{if $\type(q,n)=$ (H) and $k=w=n/2$ and $i<k$},\\ 1&\text{otherwise.}\end{array}\right.
\end{equation}
Using once more \cref{cor:SigmaDotAsymp} we obtain
\[
   \sigma_q(d-1,i)\delta_q(n,k,i)q^{\binom{i}{2}}\sim \hat{c} q^{i(d-n+k-2)+k(n-k)-\binom{k+1}{2}-\binom{i}{2}}
\]
for some positive constant~$\hat{c}$ (which equals~$c$ whenever $i<(d-1)/2$).
Note that since $d-1\geq 3$, no additional factor~$2$ for $\sigma_q(d-2,1)$ arises; see \cref{cor:SigmaDotAsymp}.
Using that $d-n+k-2<0$, one easily verifies that on the interval $[1,\infty)$ the function $f(x)=x(d-n+k-2)-x(x-1)/2$ attains its maximum at $x=1$.
This implies that the term for $i=1$ determines the asymptotics of the sum in~\eqref{sum1}.
Taking the constant into account we obtain from \eqref{sum1}--\eqref{e-deltaAsymp2}
\begin{equation}\label{e-AsympFirstSum}
    \sum_{j=1}^L |\Gamma_j|  \sim \alpha\binom{n}{d-1} q^{d+k-n-2+k(n-k)-\binom{k+1}{2}},
\end{equation}
where
\begin{equation}\label{e-alpha}
    \alpha=\left\{\begin{array}{cl}2 &\text{if $k=n/2$ and $\type(q,n)=$ (H)},\\ 1&\text{otherwise.}\end{array}\right.
\end{equation}
3) We derive an upper bound for the remaining terms in \eqref{e-inclexcl}.
Let $|J|\geq 2$.
Then $J$ contains a subset of size~2.
Without loss of generality we may assume $\{1,2\}\subseteq J$. Then $|\Gamma_J|\leq |\Gamma_{\{1,2\}}|$ and it will suffice to derive an upper bound for
$|\Gamma_{\{1,2\}}|$.
Let~$A,\,A_1$, and~$A_2$ be as in \cref{lem:S1S2}.
Then
$|\Gamma_{\{1,2\}}| \le a + b$, where
\begin{align*}
a &= |\{C \in \Sigma_q(n,k) \mid C \cap (A_1 \cap A_2) \neq \emptyset\}|, \\
b &= |\{C \in \Sigma_q(n,k) \mid C \cap (A_1 \setminus A_2) \neq \emptyset, \, C \cap (A_2  \setminus A_1) \neq \emptyset\}|.
\end{align*}
For $t \ge 1$, set $\zeta_q(t)= \sum_{i=0}^t \zeta_q(t,i)$, which is the number of self-orthogonal vectors in $\F_q^t$.
Thus $|A_i|=(\zeta_q(d-1)-1)/(q-1)$.
 \cref{prop:SOContainx} together with the fact that $S_1 \neq S_2$ provides us with the upper bound
\begin{align} \label{e-Ua}
   a \le |A_1 \cap A_2|\, \delta_q(n,k,1) \le U_a:=\frac{\zeta_q(d-2)-1}{q-1} \, \delta_q(n,k,1).
\end{align}
We now turn to~$b$.
First of all, we have $b \le \sum_{v \in A_1 \setminus A_2,\,w \in A_2 \setminus A_1}\gamma(v,w)$, where
\begin{equation}\label{e-gammavw}
      \gamma(v,w)= |\{C \in \Sigma_q(n,k)\mid \subspace{v,w} \le C \}|.
\end{equation}
In order to apply \cref{prop:SOContainx} we also need to take 
into account 
the all-1-vector $\textbf{1}$
when $\type(q,n)=(\NNA)$. 
For that type, we have
\begin{align}
   \sum_{v \in A_1 \setminus A_2,\;w \in A_2 \setminus A_1}\gamma(v,w)
     &=\sum_{\substack{v \in A_1 \setminus A_2,\;w \in A_2 \setminus A_1\\ \textbf{1}\in\langle v,w\rangle}}\gamma(v,w)
       +\sum_{\substack{v \in A_1 \setminus A_2,\;w \in A_2 \setminus A_1\\ \textbf{1}\not\in\langle v,w\rangle}}\gamma(v,w) \nonumber\\
     &\leq\sum_{\substack{v \in A_1 \setminus A_2,\;w \in A_2 \setminus A_1\\ \textbf{1}\in\langle v,w\rangle}}\hspace*{-1.5em}\sigma_q(n-3,k-2)
       +\sum_{\substack{v \in A_1 \setminus A_2,\;w \in A_2 \setminus A_1\\ \textbf{1}\not\in\langle v,w\rangle}}\hspace*{-1.5em}\delta_q(n,k,2).
       \label{e-Boundb}
\end{align}
Consider the first sum.
By \cref{lem:S1S2} there are at most $q^\ell$ pairs $(v,w)$ such that $\textbf{1}\in\subspace{v,w}$, where
$\ell=2d-n-3$ if $S_1\cap S_2\not=\emptyset$ and $\ell=0$ otherwise.
Thus the first sum in \eqref{e-Boundb} is upper bounded by $q^{\ell}\sigma_q(n-3,k-2)$.
For the second sum it will suffice to use the upper bound $((\zeta_q(d-1)-1)/(q-1))^2\delta_q(n,k,2)$, which does not even make use of the
restriction $\textbf{1}\not\in\langle v,w\rangle$.
With this generous upper bound (and by using again Proposition~\ref{prop:SOContainx}) we obtain $b\leq U'_b+U_b$ for each type of $(q,n)$, where
\begin{equation}\label{e-Ub}
    U'_b:=q^{\ell}\sigma_q(n-3,k-2)\ \text{ and }\ U_b=\Big(\frac{\zeta_q(d-1)-1}{q-1}\Big)^2\delta_q(n,k,2).
\end{equation}
Now we have
\begin{equation}\label{e-OtherTerms}
     |\Gamma_J| \le U_a+U'_b+U_b\ \text{ for all $J$ with $|J|\geq2$.}
\end{equation}
4) We turn to the asymptotics of~\eqref{e-OtherTerms} as $q\rightarrow\infty$.
Using \eqref{e-deltaAsymp1} together with \cref{prop:AsympZeta} and \cref{cor:SigmaDotAsymp} we obtain
\begin{align}
   U_a&\sim c_1 q^{(k-1)(n-k-1)-\binom{k}{2}+d-4},  \label{e-Uasim}\\
   U_b &\sim c_2  q^{(k-2)(n-k-2)-\binom{k-1}{2}+2d-6},\label{e-Ubsim}\\
    U'_b &\sim c_3  q^{(k-2)(n-k-1)-\binom{k-1}{2}+\ell},\label{e-U'bsim}
\end{align}
for some positive constant~$c_1$ and non-negative constants~$c_2,c_3$ (where we use the convention $f\sim 0$ if $f=0$).
Recall that if $|S_1\cap S_2|=0$, then $\ell=0$. In that case we have $2(d-1)=n$ and thus $k\leq n-d=n/2-1$.
If $|S_1\cap S_2|>0$, then $\ell=2d-n-3$ and we have the general assumption $k\leq n-d$.
Using these facts one verifies that the asymptotics of $U_a$ always dominates that of~$U'_b$ and either dominates or is equal to that of $U_b$.
Now  \eqref{e-OtherTerms} tells us that
\begin{equation}\label{e-GammaGrowth}
    \text{$|\Gamma_J|$ does not grow faster than $\smash{q^{(k-1)(n-k-1)-\binom{k}{2}+d-4}}$ whenever $|J|\geq2$.}
\end{equation}
5) Now we can discuss the asymptotics of~\eqref{sum1}.
It is easy to verify that the asymptotics of \eqref{e-AsympFirstSum} dominates the one in \eqref{e-GammaGrowth} and therefore
\begin{equation}\label{e-FinalAsymp}
   |\{C \in \Sigma_q(n,k) \mid d(C) \le d-1\}|\sim\alpha\binom{n}{d-1} q^{d+k-n-2+k(n-k)-\binom{k+1}{2}}
\end{equation}
with~$\alpha$ as in \eqref{e-alpha}. Dividing by $\sigma_q(n,k)$ we finally arrive at the desired statement.
\end{proof}

It remains to consider the case where $k=n-d+1$, i.e., the case of MDS self-orthogonal codes.
In that situation, the above proof strategy gives us only partial results. 
The next theorem refines the argument of 
Theorem~\ref{thm:asympdense}
to treat the remaining case.

\begin{theorem}\label{thm:asympdense2}
Let $4 \le d < n$ and $k=n-d+1\leq w(\F_q^n,\,\cdot\,)$.
\begin{alphalist}
\item Let $k<n/2$. Then
\[
  \frac{|\{C \in \Sigma_q(n,k) \mid d(C) \le d-1\}|}{\sigma_q(n,k)} \sim \binom{n}{k} q^{-1} \quad \text{ as $q \rightarrow \infty$}.
\]
\item Let $k=n/2$. Then
     \[
       \binom{n}{d-2}q^{-2}\leq \limsup_{q\rightarrow\infty}\frac{|\{C \in \Sigma_q(n,k) \mid d(C) \le d-1\}|}{\sigma_q(n,k)}\leq \binom{n}{d-1}q^{-1}.
\]
\end{alphalist}
\end{theorem}

\begin{proof}
We now revisit the proof of Theorem~\ref{thm:asympdense} and adjust it to that case where $k=n-d+1$.
The only place where we made use of $k\leq n-d$ is at the end of Step~4), where the assumption guarantees that the asymptotics of $U_a$ provide an upper bound for the asymptotics of~$|\Gamma_{J}|$ if $|J|\geq2$. 
Let now $k=n-d+1$. 
Then the asymptotics of~$U_a$, $U_b$, and  $\Psi:=\sum_{j=1}^L|\Gamma_j|$ in~\eqref{e-AsympFirstSum} read as follows:
\[
   U_a\sim c_1q^{k(n-k)-\frac{k^2+k+4}{2}},\quad \ U_b\sim c_2q^{k(n-k)-\frac{k^2+k+2}{2}},\quad \
   \Psi\sim\alpha\binom{n}{d-1}q^{k(n-k)-\frac{k^2+k+2}{2}}.
\]
As for $U'_b$ recall that $\ell=2d-n-3$ if $S_1\cap S_2\not=\emptyset$ and $\ell=0$ otherwise.
This leads to 
\[
     U'_b\sim\left\{\begin{array}{cl}c_3q^{k(n-k)-\frac{k^2-k+2n}{2}}&\text{if }S_1\cap S_2\not=\emptyset,\\[.4ex]
           c_3q^{k(n-k)-\frac{k^2-5k+4n-2}{2}}&\text{if }S_1\cap S_2=\emptyset.\end{array}\right.
\]
Clearly, the asymptotics of~$U_b$ dominates that of~$U_a$. 
Moreover, it is easily seen that it dominates that of~$U'_b$ in all cases except for $n=4$ and $k=2$, in which case $U_b$ and $U'_b$ have the same asymptotics.
\\[.5ex]
(b) Let $k=n-d+1=n/2$. The above discussion shows that the asymptotics of  \eqref{e-inclexcl} is at most that of the first term, which is~$\Psi$.
Hence
\[
   T:=\limsup_{q\rightarrow\infty}\frac{|\{C \in \Sigma_q(n,k) \mid d(C) \le d-1\}|}{\sigma_q(n,k)}\leq \binom{n}{d-1}q^{-1}.
\]
On the other hand, the set $\{C \in \Sigma_q(n,k) \mid d(C) \le d-1\}$ contains $\{C \in \Sigma_q(n,k) \mid d(C) \le d-2\}$. Replacing in \cref{thm:asympdense} the distance~$d$ by $\hat{d}=d-1$ and noticing that $k= n-\hat{d}$, that theorem implies
\[
    T\geq
    \lim_{q\rightarrow\infty}\frac{|\{C \in \Sigma_q(n,k) \mid d(C) \le \hat{d}-1\}|}{\sigma_q(n,k)}=\binom{n}{\hat{d}-1}q^{-2}.
\]
(a)
Let now $k=n-d+1<n/2$. Then $d-1>n/2$ and thus $S_1\cap S_2\not=\emptyset$. In this case the asymptotics of  $U'_b$ (and of $U_a$) 
is dominated by that of~$\Psi$.
Thus it suffices to derive a refinement of $U_b$, which is an upper bound for the second sum in \eqref{e-Boundb}. 
We will do so by taking into account that $\gamma(v,w)=0$ if $\subspace{v,w}$ is not self-orthogonal and thus will only count the pairs 
$(v,w)$ such that the subspace $\subspace{v,w}$ is self-orthogonal.
Set $B:= \sum_{(v,w)\in\mV}\gamma(v,w)$, where 
\[
      \mV=\{(v,w)\mid  v \in A_1 \setminus A_2,\,w \in A_2 \setminus A_1,\,\subspace{v,w}\in\Sigma_q(n,2),\,\textbf{1}\not\in\subspace{v,w}\}.
\]
We compute
\begin{equation}\label{e-NewBound-b}
  B= \sum_{\substack{(v,w)\in\mV\\ \wt(v)<d-1\text{ or }\wt(w)<d-1}}\hspace*{-3em}\gamma(v,w)\quad
  +\sum_{\substack{(v,w)\in\mV\\ \wt(v)=\wt(w)=d-1}}\hspace*{-1.5em}\gamma(v,w).
\end{equation}
The first summand is upper bounded by
\[
    \frac{\zeta_q(d-2)-1}{q-1}\,\frac{\zeta_q(d-1)-1}{q-1}\delta_q(n,k,2)\sim q^{d-4+d-3+(k-2)(n-k-2)-\binom{k-1}{2}}=q^{k(n-k)-\frac{k^2+k+4}{2}},
\]
and it remains to consider the second term.
Recall the sets $S_1$ and $S_2$ intersect nontrivially.
Without loss of generality let $S_1\cap S_2=\{1,\ldots,\hat{\ell}\}$ and $S_2=\{1,\ldots,d-1\}$.
We have
\[
   |\{v\in A_1\mid \wt(v)=d-1\}|=\frac{\zeta_q(d-1,d-1)-1}{q-1}
\]
and we want to upper bound the number of vectors
$w\in A_2$ of weight $d-1$ such that the subspace $\subspace{v,w}$ is self-orthogonal.
We need~$(w_1,\ldots,w_{\hat{\ell}})$ to be orthogonal to $(v_1,\ldots,v_{\hat{\ell}})$.
There exist $\smash{q^{\hat{\ell}-1}}$ such vectors $\smash{(w_1,\ldots,w_{\hat{\ell}})}$.
Each such vector can be arbitrarily extended to a vector $\smash{(w_1,\ldots,w_{d-2})}$ in $\smash{q^{d-2-\hat{\ell}}}$ ways, and then there are at most 2 ways to
extend it to a self-orthogonal vector with support contained in~$S_2=\{1,\ldots,d-1\}$.
Taking scalar multiples into account (and disregarding the weight of~$w$), we arrive at
\[
   |\{w\in A_2\mid \subspace{v,w}\text{ self-orthogonal}\}|\leq 2\frac{q^{d-3}-1}{q-1}\sim 2 q^{d-4}.
\]
All of this tells us that the second term in \eqref{e-NewBound-b} is upper bounded by
\[
   2\frac{\zeta_q(d-1,d-1)-1}{q-1}\,\frac{q^{d-3}-1}{q-1}\delta_q(n,k,2)\sim 2q^{d-3+d-4+(k-2)(n-k-2)-\binom{k-1}{2}}=2q^{k(n-k)-\frac{k^2+k+4}{2}}.
\]
Now we can conclude that the asymptotics of the right hand side in \eqref{e-inclexcl} is again dominated by that of the first sum, which in turn
is given in \eqref{e-AsympFirstSum}. As in the previous proof we arrive at~\eqref{e-FinalAsymp}, which leads to the stated result.
\end{proof}

As an immediate corollary of the previous two theorems (and of the famous Singleton bound), we obtain that self-orthogonal codes with good distance properties are dense within the set of self-orthogonal codes sharing the same dimension. More precisely, the following hold.

\begin{corollary}
Let $4 \le d<n$ and $1 \le k \le \min\{n-d+1,w(\F_q^n)\}$.
Then
\[
\frac{|\{C \in \Sigma_q(n,k) \mid d(C) \le d-1\}|}{\sigma_q(n,k)} \in O \left(q^{d+k-n-2} \right) \quad \mbox{ as $q \rightarrow \infty$}.
\]
In particular, a uniformly random $k$-dimensional self-orthogonal code $C \le \F_q^n$ has minimum distance $n-k+1$ with probability approaching 1 as $q \rightarrow \infty$.
\end{corollary}

We conclude this paper by mentioning that with the same approach as above (and in fact without requiring the tedious case-by-case analyses) one can prove that ``unrestricted'' codes follow the same asymptotics as self-orthogonal ones in \cref{thm:asympdense}.

\begin{theorem}
Let $2 \le d \le n$ and $1 \le k \le n-d+1$. Then
\[
\frac{|\{C \le \F_q^n \mid \dim(C)=k, \,  d(C) \le d-1\}|}{\GaussianD{n}{k}_q} \sim \binom{n}{d-1} q^{d+k-n-2} \quad \mbox{ as $q \rightarrow \infty$}.
\]
\end{theorem}

Together with the results from \cref{sec:5} we conclude that, with respect to distance properties and asymptotic density, self-orthogonal codes
behave similarly to ``unrestricted'' codes, despite being a sparse family within the set of codes sharing a given dimension.

\newpage

\bibliographystyle{abbrv}
\bibliography{bilin.bib}

\end{document}